\documentclass{amsart}
\usepackage{amsmath,amssymb,amsthm,amscd,verbatim,graphicx,color}
\usepackage[mathscr]{eucal}
\usepackage{lineno, enumerate,color}
\usepackage{tikz}

\theoremstyle{plain}
\newtheorem{thm}{Theorem}[section]
\newtheorem{lem}[thm]{Lemma}
\newtheorem{cor}[thm]{Corollary}
\newtheorem{pro}[thm]{Proposition}

\theoremstyle{definition}

\newtheorem{eg}[thm]{Example}
\newtheorem{defn}[thm]{Definition}

\newtheorem{remark}[thm]{Remark}
\newtheorem{notation}[thm]{Notation}

\makeatletter

\newcommand{\Rmnum}[1]{\expandafter\@slowromancap\romannumeral #1@}
\makeatother

\newcommand{\cat}[1]{\boldsymbol{\mathscr{#1}}}

\newcommand{\NAE}{\mathsf{NAE}}

\newcommand{\diag}{\operatorname{diag}}

\DeclareSymbolFont{bbold}{U}{bbold}{m}{n}
\DeclareSymbolFontAlphabet{\mathbbold}{bbold}

\newcommand{\PWb}[1]{{\boldsymbol{\raise2pt\hbox{\Large$\wp$}}(#1)}}
\newcommand{\PW}[1]{{{\raise2pt\hbox{\Large$\wp$}}(#1)}}

\DeclareMathOperator{\NTriv}{NonTriv}
\DeclareMathOperator{\Con}{Con}

\DeclareMathOperator{\id}{id}

\DeclareMathOperator{\CSP}{CSP}

\DeclareMathOperator{\Cols}{Cols}
\DeclareMathOperator{\Inv}{Inv}
\DeclareMathOperator{\Pol}{Pol}
\DeclareMathOperator{\pPol}{pPol}
\DeclareMathOperator{\dom}{dom}

\DeclareMathOperator{\SEP}{SEP}

\DeclareMathOperator{\Robust}{Robust}
\DeclareMathOperator{\Prob}{Prob}
\DeclareMathOperator{\SAT}{SAT}

\DeclareMathOperator{\II}{II}
\DeclareMathOperator{\I}{I}
\DeclareMathOperator{\IV}{IV}
\DeclareMathOperator{\IE}{IE}
\DeclareMathOperator{\IL}{IL}

\DeclareMathOperator{\IN}{IN}

\DeclareMathOperator{\GAP}{GAP}
\DeclareMathOperator{\Y}{Y}
\DeclareMathOperator{\N}{N}

\DeclareMathOperator{\EQUIV}{EQUIV}
\DeclareMathOperator{\IMPL}{IMPL}

\newcommand{\chisub}[1]{{\raise .15em \hbox{$\chi$}}_{#1}}

\newcommand{\Rob}{\textup{Rob}}
\newcommand{\up}[1]{\textup{#1}}

\newcommand{\rest}[1]{{\upharpoonright}_{#1}}

\newcommand{\onto}{\twoheadrightarrow}

\newcommand{\dflat}{\flat\kern-1.4pt\flat}

\usetikzlibrary{arrows,calc}

\tikzset{%
element/.style={draw, shape=circle, fill=white, inner sep=1.4pt},
clone/.style={draw, shape=circle, fill=white, inner sep=0pt, minimum size=18pt},
order/.style={line width=0.5pt},
inclusion/.style={line width=0.5pt},
arrow/.style={->, line width=0.5pt, >=to},
map/.style={->, shorten >=5pt, shorten <=5pt, >=stealth'},
unary/.style={->, shorten >=2pt, shorten <=2pt, >=stealth'},
loopy/.style={->, shorten >=2pt, shorten <=2pt, >=stealth, min distance=150pt},
auto}

\begin{document}

\title[Gap Theorems for Robust Satisfiability]{Gap Theorems for Robust Satisfiability: Boolean CSPs and Beyond$^*$
}\thanks{{}$^*$This article is a heavily expanded version of an extended abstract submitted to the 2016 ISAAC proceedings, and differs by the inclusion of full proofs and most results of Section~\ref{sec:algebra}.}

\author{Lucy Ham}
\address{Department of Mathematics and Statistics\\ La Trobe University\\ Victoria  3086\\
Australia} \email{leham@students.latrobe.edu.au}


\begin{abstract}
A computational problem exhibits a ``gap property" when there is no tractable boundary between two disjoint sets of instances.  We establish a Gap Trichotomy Theorem for a family of constraint problem variants, completely classifying the complexity of possible ${\bf NP}$-hard gaps in the case of Boolean domains.  As a consequence, we obtain a number of dichotomies for the complexity of specific variants of the constraint satisfaction problem: all are either polynomial-time tractable or $\mathbf{NP}$-complete. Schaefer's original dichotomy for $\SAT$ variants is a notable particular case.  

Universal algebraic methods have been central to recent efforts in classifying the complexity of constraint satisfaction problems.  A second contribution of the article is to develop aspects of the algebraic approach in the context of a number of variants of the constraint satisfaction problem.  In particular, this allows us to lift our results on Boolean domains to many templates on non-Boolean domains.
\end{abstract}

\maketitle

\section{Introduction}
Constraint Satisfaction Problems (CSPs) occur widely in practice, both as natural problems, and as an underlying framework for constraint programming; see a text such as Tsang~\cite{tsa}.
When the template is restricted to some fixed finite domain, these problems still cover many important practical problems as well as providing an important framework for theoretical considerations in computational complexity.
In the case of Boolean (2-element) domains, constraint problems coincide with the $\SAT$ variants examined by Schaefer~\cite{Schaef}. In his paper, Schaefer proved a famous dichotomy: he showed that the complexity of CSPs over a fixed Boolean constraint language is either decidable in polynomial time or is NP-complete. Since Schaefer's seminal contribution, there have been enormous advances toward a more general dichotomy for constraint satisfaction problems on non-Boolean domains. In~\cite{fedvar},  Feder and Vardi argue that fixed template CSPs emerge as the broadest natural class for which a dichotomy might hold and propose the well-known \emph{Dichotomy Conjecture}.  Numerous extensions of Schaefer's result are now known.  The broadest of these include the case of three-element domains (Bulatov~\cite{Bulatov:2006}), List Homomorphism Problems (Bulatov~\cite{Bulatov:2011}), and the case of directed graphs without sources and sinks (Barto, Kozik and Niven~\cite{BKN}).

In addition to direct extensions of Schaefer's results, many variants of constraint satisfaction problems have been shown to experience similar dichotomies, such as counting CSPs (Bulatov~\cite{Bulatov13}), balanced CSPs (Schnoor and Schnoor~\cite{Schnoor2}) and equivalence problems (B{\"{o}}hler, Hemaspaandra, Reith and Vollmer~\cite{Bohl}).  In the present work, we explore computational complexity for variants of the CSP where instead of asking for the existence of a single solution, one asks for enough solutions to witness a range of conditions. We focus in particular on \emph{separability} and \emph{robust satisfiability}. The separation problem $\SEP$ asks if it is true that, for every pair of distinct variables $u$ and $v$, there is a solution giving $u$ a different value to $v$.  The $(k, {\mathcal F})$-robust satisfiability problem asks if every compatible partial assignment on $k$ variables extends to a full solution.  As explained in Jackson~\cite{JACK}, the $\SEP$ condition arises naturally in universal algebraic considerations, and is also closely related to problems without a backbone: problems (typically SAT variants) where no variable is forced to take some fixed value. Implicit constraints such as these are widely associated with computational difficulty; see Monasson et al.~\cite{MZKST}, Beacham and Culberson \cite{BeachamCulberson05} and Culberson and Gent~\cite{CulbersonGent01}. In the language of~\cite{BeachamCulberson05} for example, the $\SEP$ condition corresponds to the ``unfrozenness of equality''.  Robust constraint problems have been studied in a number of different contexts, including Beacham~\cite{beac} and Gottlob~\cite{gott}, with the general concept formulated in Abramsky, Gottlob and Kolaitis~\cite{AGK}. The $(k,{\mathcal F})$-robustness condition is an extension of the robustness condition of \cite{AGK}. 

A remarkable result of Gottlob~\cite{gott} states that for any $k$, the decision problem ($3k+3$)$\SAT$ experiences a striking ``gap property": subject to $\mathbf{P}\not= {\bf NP}$, there is no polynomial-time decidable class of instances containing those for which every possible partial assignment on $k$ variables extends to a satisfying solution and disjoint from the set of NO instances of the corresponding $\CSP$.  In fact, every such class is $\mathbf{NP}$-hard.  Abramsky, Gottlob and Kolaitis~\cite{AGK} and then Jackson~\cite{JACK} showed that gaps of this style are also to be found for some other well-known $\mathbf{NP}$-complete problems, including $3 \SAT$, $\mathsf{G}3\mathsf{C}$, $\NAE3\SAT$, and positive $1$-in-$3$ $\SAT$.  We investigate computational gap theorems in the Boolean case, establishing a Gap Trichotomy Theorem (Theorem \ref{thm:mainGAP}) that provides dichotomies for an entire family of computational problems, including $\SEP$ and $(2,\mathcal{F})$-robust satisfiability, with all intractable cases exhibiting a gap property. In particular, we completely characterise the tractability of $\SEP$ and $(2,\mathcal{F})$-robust satisfiability in the case of Boolean domains, as well recovering Schaefer's Dichotomy Theorem in case of core relational structures. 


A pivotal development in the classification of fixed template $\CSP$ complexity, including the aforementioned extensions of Schaefer's dichotomy for Boolean CSPs, was the introduction of universal algebraic methods; starting with the work of Jeavons \cite{Jeav98}, Jeavons, Cohen, Gyssens \cite{JCG97}, with the full framework presented in Bulatov, Jeavons, Krokhin \cite{BJK}.  The algebraic method concerns the analysis of ``polymorphisms'' (see Definition \ref{defn:poly} below) of the template.  For some computational problems, polymorphism analysis appears too coarse; this is true for $\SEP$ and robust satisfiability, as well as other variants of the $\CSP$ considered in \cite{Bohl} and \cite{Schnoor2}.  In these cases, it is necessary to move to partial polymorphisms.  We make steps toward a more extensive algebraic approach for constraint-related problems amenable to partial polymorphism analysis, by showing that basic universal algebraic methods developed for $\CSP$s can also be established in this setting.  In particular, we give theorems in the context of a number of variants of the $\CSP$, including $\SEP$ and $(k,\mathcal{F})$-robust satisfiability, which serve to lift hardness results from Boolean domains to non-Boolean domains.


\subsection{Organisation of article}
The present article is comprised of eight sections: Section~\ref{sec:SepRob} introduces four computational problems that will be of primary interest. Section~\ref{sec:Back} covers the minimum required background for the full statement of our first main result.  In Section~\ref{sec:mainresults}, we give the necessary notion of a gap property and present the Gap Trichotomy Theorem, as well as numerous dichotomy theorems that are immediate consequences. Section~\ref{sec:weakbase} covers preliminary concepts required for the main arguments; these relate to weak co-clones and strong partial clones.  Sections~\ref{sec:Gaps} and \ref{sec:tractable} are dedicated to establishing the Gap Trichotomy Theorem; more specifically, Section~\ref{sec:Gaps}  deals with the intractable cases, and Section \ref{sec:tractable} with the tractable cases.   In Section~\ref{sec:algebra}, we give further development of the algebraic approach to the study of constraint-related problems amenable to partial polymorphism analysis and obtain results for the non-Boolean case. 


\section{Separation and robust satisfiability}\label{sec:SepRob}
We begin by introducing four computational problems that will be of primary focus in the present article.  
\begin{defn}
Let $\Gamma$ be a set of relation symbols, each with an associated finite arity.  A \emph{template} is a pair ${\mathbb A}=\langle A;\Gamma^{\mathbb A} \rangle$ consisting of a \emph{finite} set $A$ together with an interpretation of each $n$-ary relation symbol $r\in \Gamma$ as a subset $r^{\mathbb A}$ of ${A}^n$. The set $\Gamma^{\mathbb A}=\{r^{\mathbb A}\ |\ r\in \Gamma\}$ is often referred to as a \emph{constraint language over domain} $A$.  
\end{defn}
When the context is clear, we will blur the distinction between a constraint language and the corresponding set of relation symbols by omitting the superscript~${\mathbb A}$.

\begin{defn}
Let  ${\mathbb A}=\langle A;\Gamma^{\mathbb A} \rangle$  be a template.  We define a \emph{constraint instance for $\Gamma^{\mathbb A}$} or a \emph{$\Gamma^{\mathbb A}$-instance} to be a triple $I=(V;A;{\mathcal C})$ consisting of 
\begin{itemize}
\item a set of variables $V$, 
\item the domain set $A$, and 
\item a set of constraints $\mathcal C$.  
\end{itemize}
Each constraint $c \in {\mathcal C}$ is a pair $\langle s, r^{\mathbb A} \rangle$, where $r^{\mathbb A}$ is a $k$-ary relation in ${\Gamma}^{\mathbb A}$ and $s=(v_1, \dots, v_k)$ is a $k$-tuple involving variables from $V$. We define a \emph{solution} of $I$ to be any assignment $\phi \colon V \to A$ satisfying $(\phi(v_1), \dots, \phi(v_k))\in r^{\mathbb A}$, for each $c=\langle (v_1, \dots, v_k), r^{\mathbb A}\rangle$ in ${\mathcal C}$.  
\\[1ex]
\noindent\fbox{
\parbox{0.95\textwidth}{\emph{Constraint satisfaction problem $\CSP({\mathbb A})$ over template ${\mathbb A}$.}\\
Instance: a $\Gamma^{\mathbb A}$-instance $I$.\\
Question: is there a solution of $I$?
}

}\\[1ex] 
\noindent\fbox{
\parbox{0.95\textwidth}{\emph{Nontrivial satisfaction problem $\CSP_{\NTriv}({\mathbb A})$ over template ${\mathbb A}$.}\\
Instance: a $\Gamma^{\mathbb A}$-instance $I$.\\
Question: is there a nonconstant solution of $I$?
}
}\\[1ex] 
\noindent\fbox{
\parbox{0.95\textwidth}{\emph{Separation problem $\SEP({\mathbb A})$ over template ${\mathbb A}$}.\\
Instance: a $\Gamma^{\mathbb A}$-instance $I$.\\
Question: for every pair $\{v_1, v_2\}$ of distinct variables in $V$, is there is a solution $\phi\colon V \to A$ of $I$ such that $\phi(v_1)\not= \phi(v_2)$?
}
}

\end{defn}
Note that every YES instance of a $\SEP$ problem is a YES instance of the corresponding $\CSP$ problem. 

From an algebraic perspective, a constraint satisfaction problem over template $\mathbb A=\langle A;\Gamma^{\mathbb A} \rangle$ coincides with the homomorphism problem for the relational structure~${\mathbb A}$: an instance $I=(V;A;{\mathcal C})$ becomes a $\Gamma$-structure ${\mathbb I}=\langle V; \Gamma^{\mathbb I}\rangle$, by letting $r^{\mathbb I}=\{(v_1, \dots, v_k)\ |\ \langle (v_1, \dots, v_k), r^{\mathbb A}\rangle\in{\mathcal C}\}$, for each $k$-ary relation $r\in\Gamma$. If we let $\CSP(\mathbb A)$ be the class of all finite $\Gamma$-structures admitting a homomorphism into ${\mathbb A}$, then the constraint satisfaction problem over ${\mathbb A}$ is equivalent to deciding membership in this class. In a similar way, we can view a separation problem $\SEP({\mathbb A})$ as deciding membership in the class $\mathsf{RP}({\mathbb A})$, which consists of all finite $\Gamma$-structures for which there is an injective homomorphism into a finite direct power of ${\mathbb A}$.  See Jackson~\cite[\S 1]{JACK} for a more detailed discussion.

Before we can introduce the fourth computational problem of interest, we require some definitions. 
\begin{defn}
Let $R$ be a set of finitary relation symbols and let $X=\{x_i\ |\ i\in I\}$ be a set of pairwise distinct variables. 
A formula in the language of  $R$ is called an \emph{atomic formula} if it takes one of the following two forms: 
\begin{itemize}
\item $x_i = x_j$ for some $x_i, x_j\in V$;
\item $(x_1, \dots, x_k)\in r$, for some $k$-ary relation symbol $r\in R$ and some variables $x_1, \dots, x_k\in V$. 
\end{itemize}
We refer to formulas of the second form \emph{$R$-constraints} or \emph{$R$-tuples}. 
A formula in the language of $R$ is called a \emph{primitive-positive formula} (abbreviated to pp-formula) if, for some $l\in {\mathbb N}_0$ and $m,n\in {\mathbb N}$, it is of the form:
\[
(\exists w_1, \dots , w_\ell) \bigwedge_{i=1}^m \alpha_i( x_1, \dots, x_k, w_1, \dots , w_\ell), \tag{$\dagger$}\label{eqn:pp}
\]
where $w_1, \dots , w_\ell, x_1, \dots, x_k$ are distinct variables, and $\alpha_i( x_1, \dots, x_k, w_1, \dots , w_\ell)$ is an atomic formula with variables amongst $\{ x_1, \dots, x_k, w_1, \dots , w_\ell\}$, for each $i\in \{1, \dots m\}$.  A quantifier-free pp-formula (corresponding to $\ell=0$) is called a \emph{conjunct-atomic formula}, and called a \emph{equality-free conjunct-atomic formula} if it is a conjunct-atomic formula involving only $R$-constraints. 
\end{defn}

\begin{defn}\label{defn:rho}
Let ${\mathbb A}=\langle A;\Gamma^{\mathbb A} \rangle$ be a template.  Let $k \in {\mathbb N}$, let $x_1, \dots, x_{k}$ be pairwise distinct variables and let $\rho(x_1, \dots, x_{k})$ be a pp-formula in the language $\Gamma$. 
The formula $\rho$ defines a $k$-ary relation $r_{\rho}^{\mathbb A}$ on $A$:
\[
(a_1, \dots, a_{k})\in r_{\rho}^{\mathbb A} \iff \rho(a_1, \dots, a_{k})\ \text{is true in ${\mathbb A}$}.  
\]
Definition~\ref{defn:rho} extends to $\Gamma^{\mathbb A}$-instances $(V;A;\mathcal{C})$ in the following way: include the constraint $\langle (v_1,\dots,v_k),r_\rho^{\mathbb A}\rangle$ if $\rho(v_1, \dots, v_{k})$ is true in $(V;A;\mathcal{C})$, that is (borrowing the notation of \eqref{eqn:pp}), if there exist variables $w_1,\dots,w_{\ell}$ in $V$ such that $\langle(v_1, \dots, v_k, w_1, \dots , w_\ell),\alpha_i^{\mathbb A}\rangle\in\mathcal{C}$ for each $i\in\{1,\dots,m\}$.  
We refer to $\langle(v_1,\dots,v_k),r_\rho^{\mathbb A}\rangle$ as the \emph{constraint defined by $\rho$ from ${\mathcal C}$}. 

Let ${\mathcal F}$ be a finite set of pp-formul{\ae} in the language of ${\Gamma}$.  We define the following constraint language over $A$: 
\[ 
\Gamma^{\mathbb A}_{\mathcal F}:=\{r_{\rho}^{\mathbb A}\ |\ \rho\in {\mathcal F}\}.
\]
We let ${\mathbb A}_{\mathcal F}=\langle A;\Gamma^{\mathbb A}_{\mathcal F} \rangle$ and for any instance $(V;A;\mathcal{C})$ of $\Gamma^{\mathbb A}$ we let $\mathcal{C}_\mathcal{F}$ be the constraints defined by $\mathcal{F}$ from $\mathcal{C}$. We sometimes refer to the constraints in $\mathcal{C}_\mathcal{F}$ as \emph{$\mathcal{F}$-constraints}. 
\end{defn}

\begin{defn}
Let $\Gamma^{\mathbb A}$ be a constraint language over a finite set $A$, let $\mathcal{F}$ be a finite set of pp-formul{\ae} in the language of $\Gamma$ and let $(V;A;{\mathcal C})$ be a $\Gamma^{\mathbb A}$-instance. For a subset $S\subseteq V$, we say that an assignment $f\colon S\to A$ is \emph{${\mathcal F}$-compatible} if it preserves all constraints in $\mathcal{C}_\mathcal{F}$ on tuples from $S$, equivalently, if $f$ is a solution of $(S;A;\mathcal{C}_{\mathcal F})$ in $\CSP({\mathbb A}_{\mathcal F})$. 
\end{defn}
In other words, if for some $\rho(x_1,\dots,x_k)\in \mathcal{F}$ and some tuple $(s_1,\dots,s_k)\in S^k$ the formula $\rho(s_1,\dots,s_k)$ is true in $(V;A;{\mathcal C})$, then $\rho(f(s_1),\dots,f(s_k))$ must be true in $\mathbb{A}$.
In the following, we let $k$ be a nonnegative integer and  ${\mathcal F}$ be a finite set of pp-formul{\ae} in the language of $\Gamma$. 
\\[1ex]
\noindent\fbox{
\parbox{0.95\textwidth}{The \emph{$(k,{\mathcal F})$-robust satisfiability problem $(k,{\mathcal F})$-$\Robust({\mathbb A})$ over template ${\mathbb A}$.}\\
Instance: a $\Gamma^{\mathbb A}$-instance $I$.\\
Question: does every ${\mathcal F}$-compatible assignment on $k$ variables extend to a solution of~$I$?
}
}

\rule{0cm}{0.4cm}

In the case where ${\mathcal F}$ consists of the set of pp-formul{\ae} defining all projections of relations in $\Gamma^{\mathbb A}$, the notion of ${\mathcal F}$-compatibility has been called ``local compatibility'' and $(k,{\mathcal F})$-robust satisfiability called ``$k$-robust satisfiability'', see Abramsky, Gottlob and Kolaitis~\cite[\S 2]{AGK}, or~\cite[\S 3]{JACK} for example.  When ${\mathcal F}$ consists of these formul{\ae}, we will use the abbreviated notation $k$-$\Robust({\mathbb A})$ to refer to the $(k,{\mathcal F})$-robust satisfiability problem over ${\mathbb A}$.  In \cite[Lemma~$3.1$]{JACK}, Jackson proposes that ${\mathcal F}$-compatibility is a natural local compatibility condition.

We are primarily concerned with computational problems over finite constraint languages, however it will be useful to have a notion of tractability for constraint problems over infinite languages (over finite domains).  The following definition is in line with Bulatov, Jeavons and Krokhin~\cite[Definition~2.7]{BJK}, though we use ``locally tractable'' in preference to ``tractable''.  
\begin{defn}
Let $A$ be a finite non-empty set and let ${\Gamma}$ be a set of relations on~$A$, and let $\Prob$ be any of the computational problems $\CSP$, $\CSP_{\NTriv}$, $\SEP$, or $(k, {\mathcal F})$-$\Robust$. 

\begin{itemize}
\item If ${\Gamma}$ is an infinite set, then we say that $\Prob(\langle A;{\Gamma}\rangle)$ is \emph{locally tractable} if, for every finite subset $R$ of $\Gamma$, the problem $\Prob(\langle A;R\rangle)$ is tractable. 
\item The problem $\Prob(\langle A; {\Gamma}\rangle)$ is $\bf{NP}$-complete if there is a finite subset $S\subseteq {\Gamma}$ for which $\Prob(\langle A; S\rangle)$ is $\bf{NP}$-complete. 
\end{itemize}
\end{defn}
We will sometimes abbreviate $\Prob(\langle A; {\Gamma}\rangle)$ to $\Prob({\Gamma})$.

\section{Background: Post's lattice}\label{sec:Back}
The precise boundaries given in the Gap Trichotomy Theorem are more easily expressed using the language clone of theory.  In this section, we introduce background material required to state the main results in full detail.  The reader familiar with clone theory may skip this section and proceed to the main results in Section~\ref{sec:mainresults}. 

\begin{defn}
Let $A$ be a non-empty set.  For $n\in{\mathbb N}$, we define an \emph{operation} of arity $n$ on $A$ to be a map $f\colon A^n\to A$.   We define 
\begin{itemize}
\item[${\mathcal R}_A$] to be the set of all non-empty, non-nullary finitary relations on $A$\textup, and
\item[${\mathcal O}_A$] to be the set of all non-empty, non-nullary finitary operations on~$A$. 
\end{itemize}
\end{defn}

\begin{defn}\label{defn:poly}
Let $A$ be a non-empty set, let $n, k\in {\mathbb N}$, let $f\colon A^n \to A$ be a $n$-ary operation and let $r$ be a $k$-ary relation on the set $A$.  We say that $f$ \emph{preserves} $r$ or $r$ is \emph{invariant under} $f$ or $f$ is a \emph{polymorphism} of $r$, if whenever $a_1=(a_{11}, \dots, a_{1n}), a_2=(a_{21}, \dots, a_{2n}), \dots, a_k=(a_{k1}, \dots, a_{kn})$ are tuples in $A^n$, then
\[ 
\left(\forall i\in \{1, \dots n\}\ (a_{1i}, a_{2i}, \dots a_{ki})\in r\right) \implies (f(a_1), f(a_2), \dots, f(a_k))\in r.
\]
If $F$ is set of operations then we say that $r$ is invariant under $F$ if $r$ is invariant under every operation in $F$. 
\end{defn}

\begin{defn}
Let $A$ be a non-empty set and let ${\mathcal C}\subseteq {\mathcal O}_A$.  Then ${\mathcal C}$ is a \emph{clone} on the set $A$ if the following two conditions hold:
\begin{enumerate}[\quad \rm(1)]
\item $\mathcal C$ contains all projection operations: that is, for all $n\in {\mathbb N}$, the $i$th projection $\pi_i\colon A^n \to A$ given by $\pi_i(x_1, \dots, x_n)=x_i$ belongs to ${\mathcal C}$\textup;
\item $\mathcal C$ is closed under compositions. 
\end{enumerate}
For a set $F$ of total operations, we can define $[F]$ to be the smallest clone containing $F$ and we refer to $[F]$ as the \emph{clone generated by} $F$.  The set $F$ is sometimes called a \emph{base} for the clone $[F]$. 
\end{defn}

\begin{defn}
Let $A$ be a non-empty set.  A subset $R$ of $ {\mathcal R}_A$ is called a \emph{co-clone} or \emph{relational clone} if it is closed under the formation of pp-definable relations.   We can define $\langle R \rangle$ to be the smallest co-clone containing $R$ and we refer to $\langle R \rangle$ as the \emph{co-clone generated by} $R$.  The set $R$ is sometimes called a \emph{base} for $\langle R \rangle$. 
\end{defn}
Recall that for any non-empty set $A$, the set of all clones (co-clones) on $A$ ordered by set inclusion forms an algebraic intersection structure. 

The sets $\PW{\mathcal O_A}$ and $\PW{\mathcal R_A}$ are complete lattices, where $\PW{-}$ is the powerset operator.  The following well-known result was first observed by Geiger \cite{Geiger}. 
\begin{thm}[\cite{Geiger}]\label{thm:Geiger}
Let $A$ be a non-empty set.  The following pair of maps $\Inv\colon \PW{\mathcal O_A} \to \PW{\mathcal R_A}$ and $\Pol\colon\PW{\mathcal R_A} \to  \PW{\mathcal P_A}$ form a Galois connection  between $\PW{\mathcal O_A}$ and $\PW{\mathcal R_A}$. 
\begin{align*}
& \Inv(F):=\{r\in \mathcal R_A \ |\ \text{r is invariant under each}\ f\in F\}\ \text{and}\\
&\Pol(R):=\{f \in \mathcal O_A \ |\ \text{f preserves each}\ r\in R\},
\end{align*}
for each $F\subseteq \mathcal O_A$ and each $R\subseteq \mathcal R_A$. 
 \end{thm}

We list some useful consequences of Theorem~\ref{thm:Geiger}.
\begin{itemize}
\item  $\Pol(\Inv(F))=[F], $for all $F\subseteq \mathcal{O}_A$.
\item $\Inv(\Pol(\Gamma))=\langle \Gamma \rangle$, for all $\Gamma\subseteq \mathcal{ R}_A$.
\item The lattice of all clones on $A$ is dually isomorphic to the lattice of all co-clones on $A$.  
\item Every clone ${\mathcal C}$ corresponds to a unique co-clone $\Inv({\mathcal C})$. 
\end{itemize}
Note that large clones correspond to small co-clones. 

A characterisation of all clones on the two-element set $\{0,1\}$ was given by Post~\cite{post} and the lattice of these ``Boolean clones'' is usually called \emph{Post's Lattice}. In fact, there is a countable infinity of clones (co-clones) in the Boolean case.   An up-set of the Boolean co-clone lattice is given in Figure~\ref{clo:dualpost}; the table included gives definitions of the shaded vertices in terms of relations invariant under basic operations.  The operations $c_0$ and $c_1$ are the constant unary functions to $0$ and~$1$, respectively, and $\neg$ is the usual negation operation on $\{0,1\}$.  The shaded co-clones play a key role in formulating the main results; see Section~\ref{sec:mainresults} for a full discussion.  We refer the reader to Schnoor~\cite[Table~$3.1$]{ISch} for a full list of bases and definitions of all Boolean co-clones.

\begin{figure}
\hspace*{-6.5cm}
\rotatebox{180}{\begin{tikzpicture}[scale=0.9] 
\footnotesize
\begin{scope}[xshift=2.1cm]

\node[draw,align=left,fill=cyan] at (2,0.7) {\rotatebox{180}{$\GAP(\Y_{{\SEP}\cap(2, {\mathcal F})}, \N_{\CSP})$}};
\node[draw,align=left,fill=green] at (2,0) {\rotatebox{180}{$\GAP(\Y_{{\SEP}\cap(2, {\mathcal F})}, \N_{\NTriv})$}};
\end{scope}

\begin{scope}
\node[clone, fill=cyan] (S) at (0,0) {\rotatebox{180}{$\II_2$}};
\node[clone, fill=green] (S0) at ($(S)+(35:1)$) {\rotatebox{180}{$\II_0$}};
\node[clone, fill=green] (S1) at ($(S)+(145:1)$) {\rotatebox{180}{$\II_1$}};
\node[clone, fill=green] (S2) at ($(S1)+(35:1)$) {\rotatebox{180}{$\II$}};
\draw[inclusion] (S) to (S0);
\draw[inclusion] (S) to (S1);
\draw[inclusion] (S0) to (S2);
\draw[inclusion] (S1) to (S2);
\end{scope}

\begin{scope}[yshift=2.4cm]
\node[clone, fill=cyan] (Sp) at (0,0.15) {\rotatebox{180}{$\IN_2$}};
\node[clone, fill=green] (Sp0) at ($(Sp)+(0,0.75)$) {\rotatebox{180}{$\IN$}};
\draw[inclusion] (Sp) to (Sp0);

\draw[inclusion, bend angle=32] (S) to [bend left] (Sp);
\draw[inclusion, bend angle=30] (S2) to [bend left] (Sp0);
\end{scope}

\begin{scope}[yshift=4cm, xshift=-2.5cm]
\node[clone] (J) at (0,0){\rotatebox{180} {$\IV_2$}};
\node[clone] (J0) at ($(J)+(40:1)$){\rotatebox{180} {$\IV_0$}};
\node[clone] (J1) at ($(J)+(140:1)$) {\rotatebox{180}{$\IV_1$}};
\node[clone] (J2) at ($(J1)+(40:1)$){\rotatebox{180} {$\IV$}};
\draw[inclusion] (J) to (J0);
\draw[inclusion] (J) to (J1);
\draw[inclusion] (J0) to (J2);
\draw[inclusion] (J1) to (J2);

\draw[inclusion] (S) to (J);
\draw[inclusion] (S0) to (J0);
\draw[inclusion] (S1) to (J1);
\draw[inclusion] (S2) to (J2);
\end{scope}

\begin{scope}[yshift=4cm, xshift=2.5cm]
\node[clone] (M) at (0,0){\rotatebox{180} {$\IE_2$}};
\node[clone] (M0) at ($(M)+(40:1)$) {\rotatebox{180}{$\IE_0$}};
\node[clone] (M1) at ($(M)+(140:1)$) {\rotatebox{180}{$\IE_1$}};
\node[clone] (M2) at ($(M1)+(40:1)$) {\rotatebox{180}{$\IE$}};
\draw[inclusion] (M) to (M0);
\draw[inclusion] (M) to (M1);
\draw[inclusion] (M0) to (M2);
\draw[inclusion] (M1) to (M2);

\draw[inclusion] (S) to (M);
\draw[inclusion] (S0) to (M0);
\draw[inclusion] (S1) to (M1);
\draw[inclusion] (S2) to (M2);
\end{scope}

\begin{scope}[yshift=4cm]
\node[clone] (G3) at (0,0){\rotatebox{180}{$\IL_2$}};
\node[clone] (Gd) at ($(G3)+(140:1)$){\rotatebox{180}{$\IL_1$}};
\node[clone] (G) at ($(G3)+(40:1)$) {\rotatebox{180}{$\IL_0$}};
\node[clone] (G1) at ($(Gd)+(40:1)$){\rotatebox{180} {$\IL$}};
\node[clone] (G3p) at ($0.5*(G3)+0.5*(G1)$){\rotatebox{180}{$\IL_3$}};
\draw[inclusion] (G3) to (Gd);
\draw[inclusion] (G3) to (G);
\draw[inclusion] (G3) to (G3p);
\draw[inclusion] (Gd) to (G1);
\draw[inclusion] (G3p) to (G1);
\draw[inclusion] (G) to (G1);

\draw[inclusion, bend angle=30] (S) to [bend right] (G3);
\draw[inclusion, bend angle=20] (S0) to [bend right] (G);
\draw[inclusion, bend angle=20] (S1) to [bend left] (Gd);
\draw[inclusion, bend angle=35] (Sp0) to [bend left] (G1);
\draw[inclusion, bend angle=35] (Sp) to [bend right] (G3p);
\end{scope}

\rotatebox{180}{\begin{scope}\node at (6,-1) 
{
\setlength{\tabcolsep}{4.5pt}
    \begin{tabular}{ | l | c c |}
  \hline
 Co-clone & Definition & \\ \hline
    $\II_2$\rule{0cm}{0.3cm}  & all Boolean relations &   \\ \hline
        $\IN_2$\rule{0cm}{0.3cm} & $\Inv(\{\neg\})$ &\\ \hline
        $\II_0$\rule{0cm}{0.3cm}  & $\Inv(\{c_0\})$  &    \\ \hline
        $\II_1$\rule{0cm}{0.3cm}  & $\Inv(\{c_1\})$  &     \\ \hline
         $\II$\rule{0cm}{0.3cm}  &   $\Inv(\{c_0, c_1\})$  & \\ \hline
        $\IN$\rule{0cm}{0.3cm} & $\Inv(\{\neg, c_0\})$ &\\ \hline
        \end{tabular}	
};
\end{scope}}
\end{tikzpicture}}
\caption{An upset in the Boolean co-clone lattice, with a table of polymorphism definitions for the shaded co-clones; $\I{\mathcal C}$ abbreviates $\Inv({\mathcal C})$, for each Boolean clone~${\mathcal C}$.}\label{clo:dualpost}
\end{figure}
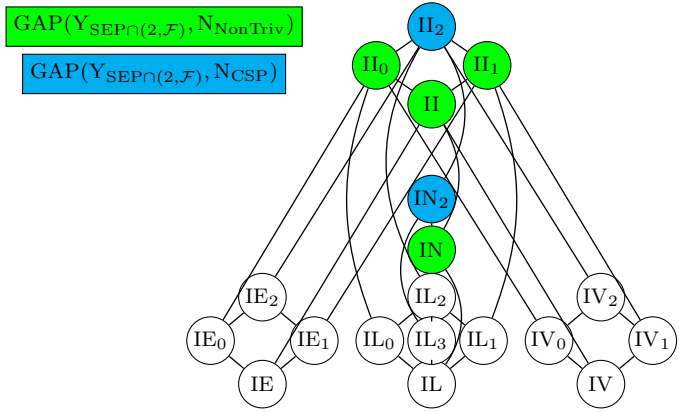

\section{Main results}\label{sec:mainresults}
The first main result presents a trichotomy of computational gap theorems for Boolean constraint languages.  As consequences, we obtain dichotomy theorems for each of the four computational problems described in Section~\ref{sec:SepRob}.  We first introduce the notion of a ``gap property'' and fix some useful notation. 

In the following, let $\Gamma$ be a finite set of relations on $\{0,1\}$ and let $P$ and $Q$ be disjoint sets of $\Gamma$-instances. 
\medskip

\noindent\fbox{
\parbox{0.95\textwidth}{
\emph{The gap property $\GAP(Q,P)$ for $\Gamma$ with respect to $Q$ and $P$}. \\
$\GAP(Q,P)$: any set of $\Gamma$-instances containing $P$ and disjoint from $Q$ has $\mathbf{NP}$-hard membership with respect to polynomial-time Karp reductions. 

}
}
\enlargethispage{\baselineskip}
\begin{notation}
Let $\Gamma$ be a finite set of relations on $\{0,1\}$, let $k\in {\mathbb N}$ and let ${\mathcal F}$ be a finite set of pp-formul{\ae} in the language of $\Gamma$. We use
\begin{itemize}
\item $\N_{\CSP}(\Gamma)$ to denote the set of NO instances for $\CSP(\Gamma)$, 
\item $\N_{\NTriv}(\Gamma)$ to denote the set of NO instances for $\CSP_{\NTriv}(\Gamma)$,
\item $\Y_{(k,{\mathcal F})}(\Gamma)$ to denote the set of YES instances for ($k,{\mathcal F}$)-$\Robust(\Gamma)$, 
\item $\Y_{k\text{-}\Rob}(\Gamma)$ to denote the set of YES instances for $k$-$\Robust(\Gamma)$,
\item $\Y_{\SEP}(\Gamma)$ to denote the set of YES instances for $\SEP(\Gamma)$,
\item $\Y_{{\SEP} \cap (k,{\mathcal F})}(\Gamma)$ to denote the set of instances in $\Y_{\SEP}(\Gamma)\cap \Y_{(k,{\mathcal F})}(\Gamma)$.
\end{itemize}
When the context refers to a specific constraint language $\Gamma$, we omit $\Gamma$ from this notation.  

\end{notation}

\begin{thm}[Gap Trichotomy Theorem]\label{thm:mainGAP}
Let $\Gamma$ be a constraint language on $\{0,1\}$.
\begin{enumerate}[\quad \rm(1)]
\item If $\langle \Gamma \rangle=\II_2$ or $\langle \Gamma \rangle=\IN_2$, then $\Gamma$ satisfies $\GAP(\N_{\CSP}, \Y_{{\SEP} \cap{(2,{\mathcal F})}})$, for some finite set ${\mathcal F}$ of pp-formul{\ae} in ${\Gamma}$. 
\item If $\langle \Gamma \rangle\in \{\II_1,\II_0 \II, \IN\}$, then $\Gamma$ satisfies $\GAP(\N_{\NTriv}, \Y_{\SEP \cap {(2,{\mathcal F})}})$, for some finite set ${\mathcal F}$ of pp-formul{\ae} in ${\Gamma}$. 
\item In all other cases, the satisfiability problem $(2, {\mathcal F})$-$\Robust(\Gamma)$ and the separation problem $\SEP(\Gamma)$ are solvable in polynomial-time, for any finite set ${\mathcal F}$ of pp-formul{\ae} in ${\Gamma}$. 
\end{enumerate}
\end{thm}
Shaded vertices in Figure~\ref{clo:dualpost} give the precise information for the Gap Trichotomy Theorem: constraint languages falling under the first item are coloured blue/dark grey in Figure \ref{clo:dualpost}.  Those falling under the second item are green/light grey.    

We now list four \emph{dichotomy} theorems that are immediate consequences of the Gap Trichotomy Theorem. The Dichotomy Theorem for $\CSP_{\NTriv}$ is a known result of Creignou and H{\'{e}}brard~\cite[Proposition $3.5, 4.7$]{Nadia}.

\begin{thm}[Dichotomy Theorem for $(2, {\mathcal F})$-$\Robust(\Gamma)$]
Let $\Gamma$ be a constraint language on $\{0,1\}$.
\begin{itemize}
\item If $\IN \subseteq \langle \Gamma \rangle$ then $(2, {\mathcal F})$-$\Robust(\Gamma)$ is $\mathbf{NP}$-complete, for some finite set ${\mathcal F}$ of pp-formul{\ae} in ${\Gamma}$.  
\item Otherwise, $(2, {\mathcal F})$-$\Robust(\Gamma)$ is solvable in polynomial-time, for all finite sets ${\mathcal F}$ of pp-formul{\ae} in ${\Gamma}$. 
\end{itemize}
\end{thm}

\begin{thm}[Dichotomy Theorem for $\SEP(\mathbb{A})$]  
Let $\Gamma$ be a constraint language on $\{0,1\}$. 
\begin{itemize}
\item If $\IN \subseteq \langle \Gamma \rangle$ then $\SEP(\Gamma)$ is $\mathbf{NP}$-complete.    
\item Otherwise, $\SEP(\Gamma)$ is solvable in polynomial-time. 
\end{itemize}
\end{thm}

\begin{thm}\cite[Dichotomy Theorem for $\CSP_{\NTriv}(\mathbb{A})$]{Nadia}
Let $\Gamma$ be a constraint language on $\{0,1\}$. 
\begin{itemize}
\item If $\IN \subseteq \langle \Gamma \rangle$ then $\CSP_{\NTriv}(\Gamma)$ is $\mathbf{NP}$-complete.  
\item Otherwise, $\CSP_{\NTriv}(\Gamma)$ is solvable in polynomial-time. 
\end{itemize}
\end{thm}
In the following it useful to recall that a finite relational structure (template) is a core if all of its unary polymorphisms are bijective. 
\subsection{Recovering Schaefer's Theorem}
We observe that the intractability cases of Schaefer's Dichotomy Theorem are recovered from Theorem~\ref{thm:mainGAP} in the case of core templates (Boolean CSPs over non-core templates are trivial).  Indeed, when $\langle \{0,1\}; \Gamma\rangle$ is a core template and $\IN \subseteq \langle \Gamma \rangle$, we have $\langle \Gamma \rangle=\IN_2$ or $\langle \Gamma \rangle=\II_2$. By Theorem~\ref{thm:mainGAP}, it follows that $\Gamma$ has gap property $\GAP(\N_{\CSP}, \Y_{{\SEP}\cap(2,{\mathcal F})})$, for some finite set ${\mathcal F}$ of pp-formal{\ae} in the language of $\Gamma$. As $\Y_{\SEP(2,{\mathcal F})}\subseteq \Y_{\CSP}$ and is disjoint from $\N_{\CSP}$, it follows that $\CSP(\Gamma)$ is $\bf{NP}$-complete. 

\begin{thm}[Schaefer's Dichotomy Theorem~\cite{Schaef}]
Let $\Gamma$ be a constraint language on $\{0,1\}$ such that $\langle \{0,1\}; \Gamma\rangle$ is a core template. 
\begin{itemize}
\item If $\IN \subseteq \langle \Gamma \rangle$ then $\CSP(\Gamma)$ is $\mathbf{NP}$-complete.  
\item Otherwise, $\CSP(\Gamma)$ is solvable in polynomial-time. 
\end{itemize}
\end{thm}

\section{Weak co-clones and strong partial clones}\label{sec:weakbase}
We now give more technical definitions that are required for the main arguments.
\begin{defn}
For $n\in{\mathbb N}$, we define a \emph{partial operation} of arity $n$ on $A$ to be a map $f\colon \dom(f)\to A$,  where $\dom(f)\subseteq A^{n}$ and define ${\mathcal P}_A$ to be the set of all non-empty, non-nullary finitary partial operations on~$A$. 
\end{defn}

\begin{defn}
Let $n, k\in {\mathbb N}$, let $f\colon \dom(f)\to A$ be a $n$-ary partial operation and let $r$ be a $k$-ary relation on the set $A$.  We say that $f$ \emph{preserves} $r$ or $r$ is \emph{invariant under} $f$ or $f$ is a \emph{partial polymorphism} of $r$, if whenever $a_1=(a_{11}, \dots, a_{1n}), a_2=(a_{21}, \dots, a_{2n}), \dots, a_k=(a_{k1}, \dots, a_{kn})$ are tuples in $\dom(f)$, then the implication of Definition~\ref{defn:poly} holds.  If $F$ is set of partial operations then we say that $r$ is invariant under $F$ if $r$ is invariant under every operation in $F$. 
\end{defn}

\begin{defn} 
Let $m, n \in {\mathbb N}$, let $f\in {\mathcal P}_A$ be $m$-ary and let $g_1,\dots, g_m \in {\mathcal P}_A$ be $n$-ary. The \emph{composition} $f(g_1,\dots, g_m)$ is an $n$-ary partial operation defined by 
\[
f(g_1,\dots, g_m)(x_1, \dots, x_n):=f(g_1(x_1,\dots, x_n), \dots, g_m(x_1, \dots, x_n)),
\]
where $\dom(f(g_1, \dots, g_m))$ is the set
\[
\biggl\{(x_1, \dots, x_n)\in \bigcap_{i=1}^m \dom(g_i)\ |\ (g_1(x_1,\dots, x_n), \dots, g_m(x_1, \dots, x_n))\in \dom(f)\biggr\}.
\]
\end{defn}

\begin{defn}
Let $f, g \in {\mathcal P}_A$.  We say that $f$ is a \emph{restriction} of $g$ or that $f$ \emph{extends} to~$g$ if $\dom(f)\subseteq\dom(g)$ and $f$ agrees with $g$ on $\dom(f)$. 
\end{defn}

\begin{defn}
Let $A$ be a non-empty set and let ${\mathcal C}\subseteq {\mathcal P}_A$.  Then ${\mathcal C}$ is a \emph{strong partial clone} if the following three conditions hold:
\begin{enumerate}[\quad \rm(1)]
\item $\mathcal C$ contains all total projection operations\textup;
\item $\mathcal C$ is closed under composition with non-empty domain\textup; 
\item $\mathcal C$ is closed under arbitrary restriction of partial operations. 
\end{enumerate}
For a set $F$ of partial operations, we can define $[F]_p$ to be the smallest strong partial clone containing $F$ and we call $[F]_p$ the \emph{strong partial clone generated} by $F$. The set $F$ is sometimes called a \emph{base} for the strong partial clone $[F]_p$. 
\end{defn}

\begin{defn}
Let $A$ be a non-empty set. A subset $R$ of ${\mathcal R}_A$ is called a \emph{weak co-clone} or \emph{weak system} if it is closed under the formation of conjunct-atomic definable relations.   We can define $\langle R \rangle_{\not\exists}$ to be the smallest weak co-clone containing $R$ and we refer to $\langle R \rangle_{\not\exists}$ as the \emph{weak co-clone generated by} $R$.  The set $R$ is sometimes called a \emph{base} for the weak system $\langle R \rangle_{\not\exists}$. 

A subset $R$ of ${\mathcal R}_A$ is called an \emph{equality-free weak co-clone} or an \emph{equality-free weak system} if it is closed under the formation of equality-free conjunct-atomic definable relations.  We define $\langle R \rangle_{\not\exists,\not=}$ to be the smallest equality-free weak co-clone containing $R$ and we refer to $\langle R \rangle_{\not\exists,\not=}$ as the \emph{equality-free weak co-clone generated by} $R$.  The set $R$ is sometimes called a \emph{base} for the equality-free weak system $\langle R \rangle_{\not\exists, \not=}$. 
\end{defn}

If we weaken the operators $\Inv$ and $\Pol$, introduced in Section~\ref{sec:Back}, to allow partial operations to be included in the definition, we obtain a refined Galois connection between the complete lattices $\PW{\mathcal P_A}$ and $\PW{\mathcal R_A}$.  The Galois connection is finer in the sense that sets of relations are further separated on the basis of their expressibility power. 
The next theorem was first observed by Romov~\cite{Romov}. It can also be obtained as a special case of a Galois connection given in Davey, Pitkethly, and Willard~\cite[Section~$1.2$]{AltEgo}.
\begin{thm}[\cite{Geiger}, \cite{Romov}]\label{thm:strgweak}
Let $A$ be a non-empty set. The following pair of maps $\Inv\colon \PW{\mathcal P_A} \to  \PW{\mathcal R_A}$ and $\Pol\colon \PW{\mathcal R_A} \to  \PW{\mathcal P_A}$ form a Galois connection  between the complete lattices $\PW{\mathcal P_A}$ and $\PW{\mathcal R_A}$.
\begin{align*}
& \Inv(F):=\{r\in \mathcal R_A \ |\ \text{r is invariant under each}\ f\in F\}\ \text{and}\\
&\pPol(R):=\{f \in \mathcal P_A \ |\ \text{f preserves each}\ r\in R\},
\end{align*}
for each $F\subseteq \mathcal P_A$ and $R\subseteq \mathcal R_A$. 
\end{thm}


In a similar way that the Galois connection between operations and relations on a non-empty set $A$ is connected to clones and co-clones,  the Galois connection between partial operations and relations is closely related to strong partial clones and weak co-clones.  In particular, sets of the form $\Inv(F)$ are precisely the weak co-clones and sets of the form $\pPol(R)$ are precisely the strong partial clones, for $F\subseteq {\mathcal P}_A$ and $R\subseteq {\mathcal R}_A$. There are analogous consequences to those listed in Section~\ref{sec:Back}, given for clones and co-clones. 

The lattice of strong partial clones on $\{0,1\}$ is complicated; it is uncountable and a complete classification appears difficult, see Sch{\"o}lzel~\cite{Scholz} for example.  Nevertheless, Post's lattice provides a useful approximation to lattice of strong partial clones in the Boolean setting: for each Boolean clone ${\mathcal C}$, is it known that the set of all strong partial clones whose total operations agree with ${\mathcal C}$ forms an interval, and there are known generators for the top and bottom elements in each of these intervals.  

In \cite{Schnoor2}, Schnoor and Schnoor give a characterisation of the largest strong partial clone corresponding to a clone~${\mathcal C}$, and use this characterisation to give a canonical construction for finding a base for the smallest element of the associated interval in the weak co-clone lattice, see \cite[Theorem~$4.11$]{Schnoor2} or \cite[Theorem~$3.9$]{ISch}.  Such bases are referred to as \emph{weak bases}, and from a complexity perspective, are crucial for establishing hardness results.  On the other hand, generators for the largest elements of these intervals, referred to as \emph{plain bases}, are in general key to proving tractability results; however an alternative approach is taken in this article.  We mention only for completeness that plain bases for all Boolean co-clones can be found in Creignou, Kolaitis and Zanuttini~\cite[Table~$2$]{CKZ}. 
\begin{defn}\label{defn:weakbase}
Let $A$ be a non-empty set, let ${\mathcal C}$ be a clone on $A$ and let $\Gamma$ be a set of finitary relations on $A$.  We call $\Gamma$ a \emph{weak base} for the co-clone $\Inv({\mathcal C})$ if ${\langle \Gamma \rangle}_{\not\exists} =\Inv({\mathcal I}_{\cup}({\mathcal C}))$, where ${\mathcal I}_{\cup}({\mathcal C})$ is the largest strong partial clone corresponding to ${\mathcal C}$.
\end{defn}

We will often present relations in a matrix form.  The representation is not unique, but it is succinct. 
\begin{defn} 
Let $r=\{a_1,\dots, a_m\}$ be a $k$-ary relation on a non-empty set $A$ with $|r|=m$. We define the matrix representation of $r$ to be the $m \times k$ matrix $M=(a_{ij})$ over $A$ with rows equal to the constraints of $r$.  
\end{defn}
\begin{defn}\label{defn:cols3}
Define $\Cols_{3}$ to be the $8$-ary relation over $\{0,1\}$:
\[
\begin{bmatrix}
 1 & 0 & 0 & 0 & 1 & 1 & 0 & 1 \\
 0 & 1 & 0 & 1 & 0 & 1 & 0 & 1 \\
 0 & 0 & 1 & 1 & 1 & 0 & 0 & 1
  \end{bmatrix}
\]
\end{defn}
\begin{defn}
Let ${\mathcal C}$ be clone on a non-empty set $A$ and let $r$ be a relation on~$A$. Define ${\mathcal C}(r)$ to be the smallest relation containing $r$ that is invariant under every operation in ${\mathcal C}$: 
\[
{\mathcal C}(r):=\bigcap \{s\ |\ s\in \Inv({\mathcal C})\ \text{and}\ r\subseteq s\}. 
\]
Following \cite{Schnoor2}, we refer to ${\mathcal C}(r)$ as the ${\mathcal C}$-closure of the relation $r$ and we say that $r$ is a ${\mathcal C}$-core of $\mathcal{C}(r)$.
\end{defn}

Using the work of Schnoor and Schnoor~\cite{Schnoor2} and Schnoor \cite[Table 3.1]{ISch}, the following construction gives weak-bases for each of the Boolean co-clones shaded in Figure~\ref{clo:dualpost}. 

\begin{pro}\up(\cite[Theorem 4.11]{Schnoor2}, \cite[Table 3.1]{ISch}\up)\label{pro:weakbase}
Let $\I{\mathcal C}$ be any of the Boolean co-clones listed in the table within Figure~\ref{clo:dualpost}. Then 
${\mathcal C}(\Cols_{3})$ is a weak-base for $\I{\mathcal C}$. 
\end{pro}

\begin{eg}\label{eg:BigWeirdoSAT}
To construct a weak base for the Boolean co-clone $\IN_{2}=\Inv(\{\neg\})$, we simply close the relation $\Cols_{3}$ under $\neg$.  Thus,
 \[
\N_2(\Cols_3)
=
 \begin{bmatrix}
 1 & 0 & 0 & 0 & 1 & 1 & 0 & 1 \\
 0 & 1 & 0 & 0 & 0 & 1 & 0 & 1 \\
 0 & 0 & 1 & 1 & 1 & 0 & 0 & 1 \\
 0 & 1 & 1 & 1 & 0 & 0 & 1 & 0 \\
 1 & 0 & 1 & 1 & 1 & 0 & 1 & 0 \\
 1 & 1 & 0 & 0 & 0 & 1& 1 & 0 \\
  \end{bmatrix}
\]
\end{eg}

Schnoor and Schnoor~\cite{Schnoor2} observe that often the more restricted closure operator $\langle -\rangle_{\not\exists,\not=}$ fits more naturally to problems than the operator $\langle -\rangle_{\not\exists}$.  This fact appears to be true of the computational problems considered in the present article.  The following irredundancy condition ensures conjunct-atomic definability without equality. 

\begin{defn} \cite[Definition 5.1]{Schnoor2}\label{defn:irrednt}
Let $r$ be a $k$-ary relation on a non-empty set $A$.  We say that the relation $r$ is \emph{$=$-redundant} if there exist a pair of equal columns in the matrix representation of $r$, and we say that $r$ is \emph{$\top$-redundant} if there is a relation $s$ of arity ${k-1}$ such that the columns of $r$ are exactly the columns of $s \times A$.  A relation $r$ is \emph{irredundant} if it is neither $=$-redundant nor $\top$-redundant.  A set $\Gamma$ of finitary relations on $A$ is \emph{irredundant} if every relation $r$ belonging to $\Gamma$ is irredundant. 
\end{defn}

\begin{remark} 
All six of the relations in Proposition \ref{pro:weakbase} are easily seen to be irredundant.
\end{remark}

Weak bases that generate not only the smallest weak-system, but also the smallest equality-free weak-system of relations generating the same co-clone, will be critical in establishing the main proofs to come.
\begin{thm}\cite[Corollary~$5.6$]{Schnoor2}\label{thm:eqfreeweak}
Let $A$ be a non-empty set, let ${\mathcal C}$ be a clone on $A$ and let $\Gamma$ be an irredundant weak base for the co-clone $\Inv({\mathcal C})$. If $\Gamma'$ is set of relations on $A$ such that $\langle \Gamma'\rangle=\Inv({\mathcal C})$, then $\langle \Gamma \rangle_{\not\exists, \not=} \subseteq \langle \Gamma' \rangle_{\not\exists, \not=}$. 
\end{thm}



\section{Towards a Trichotomy: $\mathbf{NP}$-Gaps}\label{sec:Gaps}
We now establish gap properties for relations generating the Boolean co-clones $\II_2$, $\IN_2$, $\II_0$, $\II_1$, $\II$ or $\IN$; recall exact definitions given in Figure \ref{clo:dualpost}.  Each co-clone must be considered separately, however the proofs follow the same structure: we first establish a gap property for an irredundant weak-base and then use the fact that gap properties are preserved by the $\langle -\rangle_{\not\exists, \not=}$ operator.  

We begin with three results that are crucial for establishing gap properties. The first theorem is an abridged version of \cite[Theorem~$6.1$]{JACK}, and reveals that positive $1$-in-$3$ $\SAT$ has the gap property $\GAP(\N_{\CSP}, {\mathcal K})$, where ${\mathcal K}$ is a particular subset of $\Y_{2\text{-}\Rob}$; see Theorem~\ref{thm:1in3GAP} below for the precise definition of $X$. Recall that ${\mathbbold 2}$ denotes the positive $1$-in-$3$ $\SAT$ template $\langle \{0,1\}; {+}1\text{in}3\SAT\rangle$. 

\begin{thm}\cite{JACK}\label{thm:1in3GAP}
Let ${\mathcal K}$ be the set consisting of all positive $1$-in-$3$ $\SAT$ instances $I$ with the following properties\up:
 \begin {itemize} 
\item no variable appears more than once in each constraint tuple of $I$,
\item  $I$ is $2$-robustly positive $1$-in-$3$ satisfiable.
\end{itemize}
Any set of positive $1$-in-$3$ $\SAT$ instances containing ${\mathcal K}$ and contained within the members of $\CSP(\mathbbold{2})$ has $\mathbf{NP}$-hard membership with respect to polynomial-time Karp reductions.  
 \end{thm}

The next lemma summarises the basic method employed in Abramsky, Gottlob and Kolaitis~\cite{AGK} and Jackson~\cite{JACK}.  It is essentially the definition of reduction for promise problems (see \cite[Definition~$3$]{Goldrich}, for example), but here phrased in the context of constraint problems. 

\begin{lem}\label{lem:Gap} 
Let $\Gamma$ and $\Gamma'$ be finite sets of relations on $\{0,1\}$. Let $A$ and $B$ be disjoint sets of $\Gamma$-instances  and let $X$ and $Y$ be disjoint sets of ${\Gamma'}$-instances. Further, let $\Gamma$ have the gap property $\GAP(A,B)$.  If there is a polynomial-time computable function $f\colon{\mathcal I}_{\Gamma} \to{\mathcal I}_{\Gamma'}$ satisfying:
\begin{enumerate}[\quad \rm(1)]
\item $I\in A \Rightarrow f(I)\in X$\textup,
\item $I\in B \Rightarrow f(I)\in Y$\textup,
\end{enumerate}
then $\Gamma'$ has the gap property $\GAP(X,Y)$.  In particular, $\Gamma'$ has the gap property $\GAP(f(A),f(B))$
\end{lem}

It is well known that the complexity of $\CSP(\Gamma)$ depends only on the co-clone generated by $\Gamma$, see \cite[Theorem~$3.4$]{Jeav98} or alternatively \cite[Theorem 2.16]{BJK} for a proof explicitly using pp-formul{\ae}.  We give an analogous result relating equality-free conjunct atomic definitions to reductions for $\SEP$ and $(2,{\mathcal F})$-robust satisfiability. This result is stated below as Theorem~\ref{thm:kRobustRed} and will be crucial for establishing hardness results in this section.  We first introduce a lemma that will help streamline the proof.

\begin{lem}\label{lem:CAEquivSolns}
Let $\Gamma^{\mathbb A}$ be a constraint language over a finite set $A$ and let $R^{\mathbb A}$ be a finite set of relations in $\langle \Gamma^{\mathbb A}\rangle_{\not\exists,\not=}$. There is a polynomial-time construction that transforms any instance $I=(V;A;{\mathcal C})$ of $\CSP(R^{\mathbb A})$ into an instance $I'$ of $\CSP(\Gamma^{\mathbb A})$, and moreover, the solutions of $I$ are exactly the solutions of $I'$. 
\end{lem}
\begin{proof}
The following construction is borrowed from \cite[Theorem 2.16]{BJK}.  Given an instance $I=(V;A;{\mathcal C})$ of $\CSP(\langle A;R^{\mathbb A}\rangle)$, construct an instance $I'=(V;A;{\mathcal C}')$ of $\CSP(\langle A;\Gamma^{\mathbb A}\rangle)$ in the following way.  
For each constraint $\langle (v_1, \dots, v_{\ell}) , r^{\mathbb A}\rangle\in {\mathcal C}$ with $r^{\mathbb A}$ representable by the equality-free conjunct-atomic formula:
\[
r(v_1, \dots, v_{\ell})=s_1(w_{1}^1, \dots, w_{{\ell}_1}^1) \wedge \dots \wedge s_n(w_{1}^n, \dots, w_{{\ell}_n}^n),
\]
where $\{w_{1}^1, \dots, w_{\ell_1}^1, \dots, w_{1}^n, \dots, w_{\ell_n}^n\}\subseteq \{v_1, \dots, v_\ell\}$, and $\{s_1, \dots, s_n \}\subseteq \Gamma$, we add the constraints $c_1=\langle (w_{1}^1, \dots, w_{\ell_1}^1), s_1^{\mathbb A}\rangle, \dots, c_n=\langle (w_{1}^n, \dots, w_{\ell_n}^n), s_n^{\mathbb A}\rangle$ to ${\mathcal C}$, and remove $\langle (v_1, \dots, v_\ell), r^{\mathbb A}\rangle$ from ${\mathcal C}$.  Since $r^{\mathbb A}(a_1, \dots, a_{\ell})$ is true in ${\mathbb A}$ if and only if $(a_{1}^1, \dots, a_{\ell_1}^1)\in s_1^{\mathbb A}, \dots, (a_{1}^n, \dots, a_{\ell_n}^n)\in s_n^{\mathbb A}$, for $(a_1, \dots, a_\ell)\in A^{\ell}$ and $r^{\mathbb A}\in R^{\mathbb A}$, it follows that the solutions of $I$ are precisely the solutions of $I'$. 
\end{proof}

\begin{thm}\label{thm:kRobustRed}
Let $\Gamma^{\mathbb A}$ be a constraint language over a set $A$, let $R^{\mathbb A}$ be any finite set of relations in $\langle \Gamma^{\mathbb A}\rangle_{\not\exists,\not=}$, let ${\mathcal F}$ be a finite set of pp-formul{\ae} in the language of $R$ and let $k\in{\mathbb N}$. There is a polynomial-time computable function that reduces
\begin{enumerate}[\quad \rm(1)]
\item $\CSP(R^{\mathbb A})$ to $\CSP(\Gamma^{\mathbb A})$,
\item  $\SEP(R^{\mathbb A})$ to $\SEP(\Gamma^{\mathbb A})$, and
\item $(k,{\mathcal F})$-$\Robust(R^{\mathbb A})$ to $(k, {\mathcal G})$-$\Robust(\Gamma^{\mathbb A})$, for some finite set ${\mathcal G}$ of pp-formul{\ae} in the language of $\Gamma$.
\end{enumerate}
\end{thm}

\begin{proof}
We transform an $R$-instance $I=(V;A;{\mathcal C})$ into a $\Gamma$-instance $I'=(V;A;{\mathcal C}')$ according to the construction given in the proof of Lemma~\ref{lem:CAEquivSolns}.  The reduction from $\CSP(R^{\mathbb A})$ to $\CSP(\Gamma^{\mathbb A})$ is then obtained immediately from Lemma~\ref{lem:CAEquivSolns}.  This proves~($1$). 

We now prove ($2$). Since the solutions of ${I}$ in $\CSP(R^{\mathbb A})$ are precisely the solutions of $I'$ in $\CSP(\Gamma^{\mathbb A})$, it follows that separating solutions of $I$ are exactly the separating solutions of $I'$.  Hence $I$ is a YES instance of $\SEP(R^{\mathbb A})$ if and only if $I'$ is a YES instance of $\SEP(\Gamma^{\mathbb A})$.






For ($3$), consider $r\in R$ of arity $\ell$ and abstractly expressible by an equality-free conjunct-atomic formula $r(x_1, \dots, x_\ell)$ in the language of $\Gamma$.  For each pp-formula $\rho(w_1, \dots, w_m)\in {\mathcal F}$, we construct a pp-formula $\rho_{\Gamma}(w_1, \dots, w_m)$ in the language of $\Gamma$ in the following way: replace every occurrence of an $\ell$-ary relation symbol $r$ in $\rho$ with its conjunct-atomic defining formula $r(x_1, \dots, x_\ell)$.  Let ${\mathcal G}=\{\rho_{\Gamma}\ |\ \rho \in {\mathcal F}\}$.  
Then since $\rho(a_1, \dots, a_m)$ is true in $\langle A;R^{\mathbb A}\rangle$ if and only if $\rho_{\Gamma}(a_1, \dots, a_m)$ is true in $\langle A;\Gamma^{\mathbb A}\rangle$ for $(a_1, \dots, a_m)\in A^{m}$ and $\rho(w_1, \dots, w_m)\in {\mathcal F}$, it follows that the ${\mathcal F}$-compatible assignments on $k$ variables of $I$ are exactly the ${\mathcal G}$-compatible assignments on $k$ variables of $I'$.  Thus, since the solutions of $I$ are precisely the solutions of $I'$, by Lemma~ \ref{lem:CAEquivSolns}, it then follows that $I$ is a YES instance of $(k,{\mathcal F})$-$\Robust(R^{\mathbb A})$ if and only if $I'$ is a YES instance of $(k, {\mathcal G})$-$\Robust(\Gamma^{\mathbb A})$. 
 \end{proof}
Theorem~\ref{thm:kRobustRed} holds more generally: with some caveats and proper amendment to the proof, the assumption that $R\subseteq \langle \Gamma\rangle_{\not\exists,\not=}$ can be weakened to $R\subseteq \langle \Gamma \rangle_{\not\exists}$. However, this result is not required for establishing our main theorems.  We refer the reader to Section~\ref{sec:algebra} for a proof.

For the purpose of easy reference, we extract the following consequence of Theorem~\ref{thm:kRobustRed}.  It may be useful to recall that for a constraint language $\Gamma$, $\N_{\CSP}(\Gamma)$ is set of NO instances for $\CSP(\Gamma)$, and  $\Y_{{\SEP}\cap(2, \mathcal{F})}(\Gamma)$ is the set of instances that are simultaneously YES for $\SEP(\Gamma)$ and $(2, \mathcal{F})$-$\Robust$. 
\begin{cor}\label{cor:GapRed}
Let $\Gamma$ be a set of relations on $\{0,1\}$, let $R$ be a finite set of relations in $\langle \Gamma\rangle_{\not\exists,\not=}$ and let ${\mathcal F}$ be a finite set of pp-formul{\ae} in the language of~$R$. There is a polynomial-time reduction taking 
\begin{itemize}
\item $\N_{\CSP}(R)$ to $\N_{\CSP}(\Gamma)$, and
\item $\Y_{{\SEP}\cap(2, \mathcal{F})}(R)$ to $\Y_{{\SEP}\cap(2,{\mathcal G})}(\Gamma)$, for some finite set ${\mathcal G}$ of pp-formul{\ae} in $\Gamma$.
\end{itemize}
\end{cor}




 We now give a result that shows when establishing a gap property for a set of relations generating a given co-clone, it is sufficient to establish the gap property for an irredundant weak base.  

\begin{thm}\label{thm:GAP1&2}
Let ${\mathcal C}$ be a clone on $\{0,1\}$, let $W$ be an irredundant weak base for the co-clone $\Inv({\mathcal C})$ and let $\Gamma$ be a constraint language on $\{0,1\}$ such that $\langle \Gamma \rangle=\Inv({\mathcal C})$. 
\begin{enumerate}[\rm(1)]
\item If $W$ satisfies $\GAP(\N_{\CSP},\Y_{{\SEP}\cap(2,{\mathcal F})})$, for some finite set ${\mathcal F}$ of pp-formul{\ae} in the language of $W$, then $\Gamma$ satisfies $\GAP(\N_{\CSP},\Y_{{\SEP}\cap(2,{\mathcal G})})$, for some finite set ${\mathcal G}$ of pp-formul{\ae} in the language of $\Gamma$. 
\item If $W$ satisfies $\GAP(\N_{\NTriv}, \Y_{{\SEP}\cap{(2,{\mathcal F})}})$, for some finite set ${\mathcal F}$ of pp-formul{\ae} in the language of $W$, then $\Gamma$ satisfies $\GAP(\N_{\NTriv}, \Y_{{\SEP}\cap{(2,{\mathcal G})}})$, for some finite set ${\mathcal G}$ of pp-formul{\ae} in the language of $\Gamma$.
\end{enumerate}
 \end{thm}
\begin{proof}
($1$): Since $W$ is an irredundant weak-base for the co-clone $\Inv({\mathcal C})$, by Theorem~\ref{thm:eqfreeweak} we have $W\subseteq\langle \Gamma \rangle_{\not\exists,\not=}$.  Now as $W$ satisfies $\GAP(\N_{\CSP},\Y_{{\SEP}\cap(2,{\mathcal F})})$, for some finite set ${\mathcal F}$ of pp-formul{\ae} in the language of $W$, we may apply Corollary~\ref{cor:GapRed}. It then follows from Lemma~\ref{lem:Gap}, that $\Gamma$ has the gap property $\GAP(\N_{\CSP},\Y_{{\SEP}\cap{(2,{\mathcal G})}})$, for some finite set ${\mathcal G}$ of pp-formul{\ae} in the language of $\Gamma$. 

($2$): From Theorem~\ref{thm:eqfreeweak}, we have $W \subseteq{\langle \Gamma \rangle}_{\not\exists,\not=}$, and by assumption we have that $W$ satisfies $\GAP(\N_{\NTriv}, \Y_{{\SEP}\cap{(2,{\mathcal F})}})$, for some finite set ${\mathcal F}$ of pp-formul{\ae} in the language of $W$.  The reduction given in Theorem~\ref{thm:kRobustRed} preserves not only the instances in $\Y_{{\SEP}\cap (2,\mathcal{F})}$ (subject to a change in ${\mathcal F}$), but also solutions of instances and hence it preserves instances in $\N_{\NTriv}$.  By Lemma~\ref{lem:Gap}, it then follows that $\Gamma$ satisfies $\GAP(\N_{\NTriv},\Y_{{\SEP}\cap(2,{\mathcal G})})$. 
\end{proof}

The first two subsections cover the co-clones ${\II}_2$ and ${\IN}_2$ corresponding to Statement $1$ of the Gap Trichotomy Theorem. Part ($1$) of Theorem~\ref{thm:GAP1&2} tells us we need only establish the gap property $\GAP(\N_{\CSP},\Y_{{\SEP}\cap(2,{\mathcal F})})$ for an irredundant weak base in each case.
\subsection{The Boolean co-clone ${\II}_2$}
Recall the relation $\Cols_3$ given in Definition~\ref{defn:cols3}. By Proposition~\ref{pro:weakbase}, the relation $\I_2(\Cols_3)=\Cols_3$, is a weak base for ${\II}_2$, and is easily verified to be irredundant.  We will use $\mathbb{II}_2$ to denote the template $\langle \{0,1\}; \Cols_3\rangle$ and $\II_2\text{-}\SAT$ as notation for the constraint problem $\CSP(\mathbb{II}_2)$. 

In the following, it is useful to remember that the positive $1$-in-$3$ $\SAT$ relation has gap property $\GAP(\N_{\CSP}, {\mathcal K})$, where ${\mathcal K}$ satisfies the properties given in Theorem~\ref{thm:1in3GAP}.  

\begin{pro}\label{pro:WeirdoGAP}
 The relation $\I_2(\Cols_3)$ satisfies $\GAP(\N_{\CSP}, \Y_{{\SEP}\cap{2\text{-}\Rob}})$. 
\end{pro}

\begin{proof} 
We begin by showing there is a polynomial-time Karp reduction from $\CSP({\mathbbold 2})$ to $\II_2\text{-}\SAT$.  We then show that this reduction takes instances with the properties given in Theorem~\ref{thm:1in3GAP} to instances in $\Y_{{\SEP}\cap{2\text{-}\Rob}}(\I_2(\Cols_3))$.  The result will then follow from Theorem~\ref{thm:1in3GAP} and Lemma~\ref{lem:Gap}. 

We motivate the following construction by observing that columns $4$, $5$ and $6$ of $\I_2(\Cols_3)$ are the negations of columns $1$, $2$ and $3$, respectively; refer to Defintion~\ref{defn:cols3}.  Given an instance $I=(V;\{0,1\};\mathcal C)$ of positive $1$-in-$3$ $\SAT$, we construct an instance $I^{\star}=(V^{\star};\{0,1\};\mathcal{C}^{\star})$ of $\II_2\text{-}\SAT$ in the following way. 
\begin{enumerate}[\quad \rm(i)]
\item First let $\overline{V}=\{\bar{v}\ |\ v\in V\}$ be a disjoint copy of $V$, and then construct $V^{\star}=V\cup\overline{V}\cup\{\top,\bot\}$, where $\top,\bot \not\in V\cup\overline{V}$,
\item for each constraint $\langle (x,y,z), {+}1\text{in}3\SAT\rangle$ in ${\mathcal C}$, we add the constraint \linebreak$\langle (x,y,z,\bar{x}, \bar{y}, \bar{z}, \bot, \top), \I_2(\Cols_3)\rangle$ to $\mathcal C$ and remove $\langle (x,y,z), {+}1\text{in}3\SAT \rangle$.
\end{enumerate}
Any solution $\varphi$ of $I$ in $\CSP({\mathbbold 2})$ can be extended to a solution $\varphi^{\star}$ of $I^{\star}$ in the following way.  For each $v\in V$, define $\varphi^{\star}(v):=\varphi(v)$, $\varphi^{\star}(\bar{v})=\neg\circ\varphi(v)$, $\varphi^{\star}(\bot)=0$ and $\varphi^{\star}(\top)=1$, where $\neg$ is the usual Boolean complement.  For the converse direction, observe that the projection $\pi_{\{1,2,3\}}(\I_2(\Cols_3))={+}1\text{in}3\SAT$, thus if $\psi$ is a solution of $I^{\star}$ in $\II_2\text{-}\SAT$, then the restriction $\psi\rest{V}$ is a solution of $I$ in $\CSP({\mathbbold 2})$. 
Hence we have shown that any solution $\phi$ of $I$ extends uniquely to a solution $\phi^{\star}$ of $I^{\star}$. 

Now assume that $I$ has the properties listed in Theorem~\ref{thm:1in3GAP}, that is, $I$ is $2$-robustly satisfiable and no variable appears more than once in each constraint tuple. Then $I$ has the following properties.  
\begin{enumerate}[\quad \rm(1)]
\item[($\heartsuit$)] For every pair of distinct variables $x$ and $y$ in $V$, there are solutions $\varphi_1$, $\varphi_2$ and $\varphi_3$ of $I$ satisfying 
\begin{align*}
&(\varphi_1(x), \varphi_1(y))=(0,0),\\
&(\varphi_2(x), \varphi_2(y))=(0,1), \\ 
&(\varphi_3(x), \varphi_3(y))=(1,0). 
\end{align*}
\item[($\Diamond$)] If $x$ and $y$ do not appear in a common constraint tuple, then there is a solution $\varphi_4$ of $I$ satisfying $(\varphi_4(x), \varphi_4(y))=(1,1)$. 
\end{enumerate}
We first show that $I^{\star}\in \Y_{2\text{-}\Rob}(\I_2(\Cols_3))$. Let $\{u, v\}$ be a pair of distinct variables in $V^{\star}$ and let $\alpha\colon\{u,v\}\to\{0,1\}$ be an assignment that is locally compatible for $I^{\star}$. There are four cases to consider. 

\emph{Case $1$}. Assume that $\{u, v \}\subseteq V$. 

\emph{Subcase}~$1$(a): If $u$ and $v$ appear in a common constraint tuple, then there are three locally compatible assignments to check: $\alpha$ satisfying
 \[
(\alpha(u),\alpha(v))=(0,0), (\alpha(u),\alpha(v))=(0,1),\ \text{or}\ (\alpha(u),\alpha(v))=(1,0).
\]
From ($\heartsuit$), there is a solution of $I$ extending each of these assignments, namely $\varphi_1$, $\varphi_2$ and $\varphi_3$, respectively. Using the fact that every solution of $I$ extends to a solution of $I^{\star}$, we get solutions $\varphi_1^{\star}$, $\varphi_2^{\star}$ and $\varphi_3^{\star}$ of $I^{\star}$ which extend in the required ways. 

\emph{Subcase}~$1$(b): If $u$ and $v$ do not appear in a common constraint tuple, then every possible assignment for $\{u,v\}$ is locally compatible. Now apply ($\heartsuit$) and ($\Diamond$) and the fact that every solution $I$ extends to a solution of $I^{\star}$, to get solutions of $I^{\star}$ for every possible assignment. 

\emph{Case}~$2$.  Assume that $u\in V$ and $v\in \overline{V}$.  We will consider cases only up to symmetry. 

\emph{Subcase}~$2$(a): If $v =\bar{u}$.  In this case there are two locally compatible assignments for $\{u,v\}$: $(\alpha(u),\alpha(v))=(0,1)$ or $(\alpha(u),\alpha(v))=(1,0)$.  Consider first $\alpha$ satisfying $(\alpha(u),\alpha(v))=(0,1)$. Using ($\heartsuit$), there is a solution $\varphi_2$ of $I$ extending this assignment.  Since $\varphi_2$ satisfies $\varphi_2(u)=0$, it follows that $\varphi_2^{\star}(u)=0$ and $\varphi_2^{\star}(v)=\varphi_2^{\star}(\bar{u})=\neg\circ\varphi_2(u)=1$.  Hence $\varphi_2^{\star}$ is a solution of $I^{\star}$ extending this assignment. The case for $(\alpha(u),\alpha(v))=(1,0)$ is similar.  

\emph{Subcase}~$2$(b): If $v =\bar{w}$, for some $w\in V$ with $w\not =u$, then there are two subcases to check.  

$2$(b)(i): If $u$ and $w$ (and hence $v$) appear in a common constraint tuple. There are three locally compatible assignments for $\{u,v\}$ to check in this case: $\alpha$ satisfying
\[
(\alpha(u),\alpha(v))=(0,1), (\alpha(u),\alpha(v))=(1,1),\ \text{or}\ (\alpha(u),\alpha(v))=(0,0).
\]
Consider first $(\alpha(u),\alpha(v))=(0,1)$.  From ($\heartsuit$), there is a solution $\varphi_1$ of $I$ satisfying $(\varphi_1(x), \varphi_1(y))=(0,0)$.  It then follows that $\varphi_1^{\star}(u)=0$ and $\varphi_1^{\star}(v)=\varphi_1^{\star}(\bar{w})=\neg\circ\varphi_1(w)=1$.  Hence we have found a solution of $I^{\star}$ extending this assignment. Now consider $(\alpha(u),\alpha(v))=(1,1)$. There is a solution $\varphi_3$ of $I$ satisfying $(\varphi_3(u),\varphi_3(w)=(1,0)$, guaranteed by ($\heartsuit$).  It follows that $\varphi_3^{\star}(u)=1$ and $\varphi_3^{\star}(v)=\varphi_3^{\star}(\bar{w})=1$.  In the case where $\alpha$ satisfies $(\alpha(u),\alpha(v))=(0,0)$, the solution $\varphi_2^{\star}$ gives the required extension. 

$2$(b)(ii): If $u$ and $v$ do not appear in a common constraint tuple, then every possible assignment for  $\{u,v\}$ is locally compatible.  For the assignments $(\alpha(u),\alpha(v))=(0,1)$, $(\alpha(u),\alpha(v))=(1,1)$ and $(\alpha(u),\alpha(v))=(0,0)$, we apply Subcase~$2$(b)(i). It remains to show there is a solution extending $(\alpha(u),\alpha(v))=(1,0)$. From ($\Diamond$), there is a solution $\varphi_4$ of $I$ satisfying $(\varphi_4(u), \varphi_4(w))=(1,1)$. It follows that $\varphi_4^{\star}(u)=1$ and $\varphi_4^{\star}(w)=\varphi_4^{\star}(\bar{w})=\neg\circ\varphi_4(w)=0$. 

\emph{Case}~$3$. Assume that $\{u,v\}\in \overline{V}$.  Then $u =\bar{z}$ for some $z\in V$ and $v =\bar{w}$, for some $w\in V$.  Apply Case $1$ to the pair $\{z,w\}$. 


\emph{Case}~$4$: If $v\in\{\bot,\top\}$ then there are three subcases to consider.  

\emph{Subcase}~$4$(a): If $u\in\{\bot,\top\}$, then up to symmetry the only locally compatible assignment is $(\alpha(u), \alpha(v))=(0,1)$. A solution of $I^{\star}$ extending this assignment is obtained from any solution of $I$.  

\emph{Subcase}~$4$(b): If $u\in V$, then up to symmetry in the second coordinate there are two locally compatible assignments to consider: $\alpha$ satisfying $(\alpha(u), \alpha(v))=(0,0)$ or $(\alpha(u), \alpha(v))=(1,0)$.  For the assignment $(\alpha(u), \alpha(v))=(0,0)$, use any solution $\varphi$ of $I$ satisfying $\varphi(u)=1$ guaranteed by $({\heartsuit})$; since the extension $\varphi^{\star}$ of any solution of $I$ satisfies $\varphi^{\star}(\bot)=0$.  For $(\alpha(u), \alpha(v))=(1,0)$ use any solution $\varphi$ of $I$ satisfying $\varphi(u)=1$ guaranteed by $({\heartsuit})$. 

\emph{Subcase}~$4$(b): If $u=\bar{w}$, for some $w\in V$.  In this case, we apply {Subcase}~$4$(b) to the pair $\{w, v\}$. 


To show $I^{\star}\in \Y_{\SEP}(\I_2(\Cols_3))$ observe that Cases $1$ and $2$ above show that for any pair of distinct variables $\{u, v\}$ in $V^{\star}$, there is a solution of $I^{\star}$ in $\CSP(\mathbb{II}_2)$ separating $u$ and $v$.  
\end{proof}

\begin{remark}\label{rem:bottop}
It is immediate from the proof that there is no loss of generality if the following extra condition is included in Proposition \ref{pro:WeirdoGAP} as an allowed assumption on instances: there are variables $\bot$ and $\top$, such that in every constraint tuple, the final two coordinates are $\bot$ and $\top$, in that order. 
\end{remark}
The following is an immediate consequence of Proposition~\ref{pro:WeirdoGAP} and Theorem~\ref{thm:GAP1&2}. 

\begin{thm}\label{thm:GAPII2}
Let $\Gamma$ be a constraint language on $\{0,1\}$ such that $\langle \Gamma \rangle=\II_2$. Then $\Gamma$ has the gap property $\GAP(\N_{\CSP},\Y_{{\SEP}\cap(2,{\mathcal G})})$, for some finite set ${\mathcal G}$ of pp-formul{\ae} in the language of $\Gamma$, and consequently both $(2, {\mathcal G})$-$\Robust(\Gamma)$ and $\SEP(\Gamma)$ are $\bf{NP}$-complete.
 \end{thm}



\subsection{The Boolean co-clone $\IN_2$} 
Recall that the relation $\N_2(\Cols_3)$ constructed in Example~\ref{eg:BigWeirdoSAT} is an irredundant weak base for $\IN_2$.  We will use $\mathbb{IN}_2$ to denote the corresponding template $\langle \{0,1\}; \N_2(\Cols_3)\rangle$ and $\IN_2$-$\SAT$ as notation for the constraint problem $\CSP(\mathbb{IN}_2)$. 

\begin{pro}\label{pro:BigWGAP}
The relation $\N_2(\Cols_3)$ satisfies $\GAP(\N_{\CSP}, \Y_{{\SEP}\cap{2\text{-}\Rob}})$.
\end{pro}
\begin{proof} 
We show there is polynomial-time Karp reduction from positive $1$-in-$3$ $\SAT$ to $\IN_2$-$\SAT$ that takes instances satisfying the properties listed in Theorem~\ref{thm:1in3GAP} to instances in $\Y_{{\SEP}\cap {2\text{-}\Rob}}(\N_2(\Cols_3))$. The result then follows from Theorem~\ref{thm:1in3GAP} and Lemma~\ref{lem:Gap}. 

Given an instance $I=(V;\{0,1\};\mathcal C)$ of positive $1$-in-$3$ $\SAT$, we construct an instance $I^{\#}=(V^{\#};\{0,1\};\mathcal{C}^{\#})$ of $\IN_2$-$\SAT$ in the following way: 
\begin{enumerate}[\quad \rm(i)]
\item for each variable $v\in V$, we add the variable $\bar{v}$ to $V$ (each of the variables $\bar{v}$ is a new variable not in $V$),
\item add the new variables $\top$ and $\bot$ to $V$, 
\item for each constraint $\langle (x,y,z), {+}1\text{in}3\SAT \rangle$ in ${\mathcal C}$, we add the constraints
\begin{align*}
&\langle (x,y,z,\bar{x}, \bar{y}, \bar{z}, \bot, \top), \N_2(\Cols_3)\rangle, \ \text{and}\\
&\langle(\bar{x}, \bar{y}, \bar{z}, x,y,z, \top, \bot), \N_2(\Cols_3)\rangle
\end{align*}
to $\mathcal C$ and remove $\langle (x,y,z), {+}1\text{in}3\SAT\rangle$ from ${\mathcal C}$.
\end{enumerate}
Any solution $\varphi$ of $I$ in $\CSP({\mathbbold 2})$ can be extended to a solution $\varphi^{\#}$ of $I^{\#}$ in $\II_2\text{-}\SAT$ in the same way as in Proposition~\ref{pro:WeirdoGAP}: for all $v\in V$, define $\varphi^{\#}(v):=\varphi(v)$, $\varphi^{\#}(\bar{v}):=\neg\circ\varphi(v)$, $\varphi^{\#}(\bot):=0$ and $\varphi^{\#}(\top):=1$. 
For the converse direction, we will show that for any solution $\varphi^{\#}$ of $I^{\#}$, 
either $\varphi^{\#}$ or $\neg\varphi^{\#}$ can be restricted to a solution of $I$.  Since the map $\neg:\{0,1\}\to \{0,1\}$ is an automorphism of $\mathbb{IN}_2$, it follows that $\neg\circ\varphi^{\#}$ is solution of $I^{\#}$, whenever $\varphi^{\#}$ is a solution of $I^{\#}$.  Thus, without loss of generality we may assume that $\varphi^{\#}(\bot)=0$. This implies that $\varphi^\#$ satisfies constraints of the form $\langle (x,y,z,\bar{x}, \bar{y}, \bar{z}, \bot, \top), \N_2(\Cols_3)\rangle$ into the subset $\I_2(\Cols_3)$ of $\N_2(\Cols_3)$ (that is, the first three rows of $\N_2(\Cols_3)$).  Then $\varphi^{\#}{\rest{V}}$ is a solution of $I$ for positive $1$-in-$3$ $\SAT$.

Now assume that $I$ has the properties given in Theorem~\ref{thm:1in3GAP}.
We show that $I^{\#} \in \Y_{2\text{-}\Rob}(\N_2(\Cols_3))$ by showing that every locally compatible assignment on two variables extends to a solution of $I^{\#}$. Recall the instance $I^{\star}=(V^{\star};\{0,1\};{\mathcal C}^{\star})$ constructed in Proposition~\ref{pro:WeirdoGAP}.  The following fact will be useful. 
\begin{enumerate}
\item[($\star$)] for every solution $\varphi^{\star}$ of $I^{\star}$ in $\IN_2\text{-}\SAT$, both $\varphi^{\star}$ and $\neg\circ \varphi^{\star}$ is a solution of $I^{\#}$ in $\IN_2\text{-}\SAT$. \end{enumerate}

Let $\{u,v\}$ be distinct variables in $V$ and let $\alpha\colon\{u,v\}\to\{0,1\}$ be an assignment that is locally compatible for $I^{\#}$.  There are two cases to consider. 

\emph{Case}~$1$. If $\alpha$ is locally compatible assignment for $I^{\star}$, then use the fact that $I^{\star}$ is in $\Y_{2\text{-}\Rob}(\I_2(\Cols_3))$ to extend $\alpha$ to a solution $\psi^{\star}$ of $I^{\star}$, which is also a solution of $I^{\#}$ by ($\star$). 

\emph{Case}~$2$. If $\alpha$ is not locally compatible for $I^{\star}$ then is it straightforward to check that $\neg\circ\alpha$ is a locally compatible for $I^\star$.  We then use the fact that $I^{\star}$ is in $\Y_{2\text{-}\Rob}(\I_2(\Cols_3))$ to extend $\neg \circ\alpha$ to a solution $\phi^{\star}$ of $I^{\star}$. Then $\neg\circ\phi^{\star}$ is a solution of $I^{\#}$ by ($\star$) and extends~$\alpha$.

It remains to show $I^{\#}\in \Y_{\SEP}(\N_2(\Cols_3))$.  Let $\{u, v\}$ be any pair of distinct variables in $V^{\#}$.  Then $\{u,v \}\subseteq V^{\star}$. Since $I^{\star}\in \Y_{\SEP}(\I_2(\Cols_3))$, there is a solution $\varphi^{\star}$ of $I^{\star}$ separating $u$ and $v$.  Then $\varphi^{\star}$ is a solution of $I^{\#}$ that separates $u$ and $v$, by~$(\star)$.  
\end{proof}
\begin{remark}\label{rem:bottop2}
As with Remark \ref{rem:bottop}, we may additionally add the following technical assumption to Proposition \ref{pro:BigWGAP}: there are two variables $\bot$ and $\top$ and every constraint tuple has its final two coordinates equal to either $\bot$,$\top$ or $\top$,$\bot$.
\end{remark}
The next result immediately follows from Proposition \ref{pro:BigWGAP} and Theorem~\ref{thm:GAP1&2}. 

\begin{thm}\label{thm:GAPIN2}
Let $\Gamma$ be a constraint language on $\{0,1\}$ such that $\langle \Gamma \rangle=\IN_2$. Then $\Gamma$ has the gap property $\GAP(\N_{\CSP}, \Y_{{\SEP}\cap(2,{\mathcal G})})$, for some finite set ${\mathcal G}$ of pp-formul{\ae} in the language of $\Gamma$, and consequently the problems $(2, {\mathcal G})\text{-}\Robust(\Gamma)$ and $\SEP(\Gamma)$ are $\bf{NP}$-complete. 
\end{thm}

For the remaining clones $\II_1, \II_0,\II,\IN$ corresponding to Statement 2 in the Gap Dichotomy Theorem~\ref{thm:mainGAP}, we can reuse the same fundamental construction used for $\II_2$ and $\IN_2$. In all cases, the proof proceeds as follows: to lie in $\Y_{\SEP}$ or $\Y_{(2,\mathcal{F})}$, it is necessary to have an assignment in which $\bot$ and $\top$ take different values.  We then argue that this forces solutions into $\II_2$ or $\IN_2$.  
Part ($2$) of Theorem~\ref{thm:GAP1&2} tells us we need only establish the gap property $\GAP(\N_{\NTriv}, \Y_{{\SEP}\cap{(2,{\mathcal F})}})$ for an irredundant weak base in each case.


\subsection{The Boolean co-clone $\II_1$.}
From Proposition~\ref{pro:weakbase}, the relation $\I_1(\Cols_{3})$ is a weak base for the co-clone $\II_1$, and is easily seen to be irredundant:
\[
\I_1(\Cols_{3})=
 \begin{bmatrix}
 1 & 0 & 0 & 1 & 1 & 0 & 0 & 1 \\
 0 & 1 & 0 & 1 & 0 & 1 & 0 & 1 \\
 0 & 0 & 1 & 0 & 1 & 1& 0 & 1 \\
  1 & 1 & 1 & 1& 1 & 1& 1 & 1 \\
  \end{bmatrix}
\]
We will use $\mathbb{II}_1$ to denote the template $\langle \{0,1\};\I_1(\Cols_{3})\rangle$ and $\II_1\text{-}\SAT$ to denote $\CSP(\mathbb{II}_1)$. 

Observe that the constraint satisfaction problem $\II_1\text{-}\SAT$ is trivial: every instance $I$ has a solution; specifically, the map that sends every variable of $I$ to $1$.  Recall that $\N_{\NTriv}(\I_1(\Cols_{3}))$ is the set of all instances of $\II_1\text{-}\SAT$ that admit only a trivial solution. Note that instances in $\N_{\NTriv}(\I_1(\Cols_{3}))$ are precisely the instances of $\II_1\text{-}\SAT$ with at most one solution. We establish the following gap property for $\I_1(\Cols_{3})$. 
\begin{pro}\label{pro:GAPWSAT1}
The relation $\I_1(\Cols_{3})$ satisfies $\GAP(\N_{\NTriv}, \Y_{{\SEP}\cap{2\text{-}\Rob}})$. 
\end{pro}
\begin{proof}
We show there is a polynomial-time reduction taking NO instances of the problem $\II_2\text{-}\SAT$ to instances in $\N_{\NTriv}(\I_1(\Cols_3))$ and $\Y_{{\SEP}\cap2\text{-}\Rob}(\I_2(\Cols_3))$ to $\Y_{{\SEP}\cap2\text{-}\Rob}(\I_1(\Cols_3))$. The final result will then follow from Proposition~\ref{pro:WeirdoGAP} and Lemma~\ref{lem:Gap}.

Given an instance $I$ of $\II_2\text{-}\SAT$, construct an instance $\chisub{}(I)$ of $\II_1\text{-}\SAT$ by replacing each constraint $\langle t=(x_1,\dots, x_8),\I_2(\Cols_3)\rangle$ in ${\mathcal C}$ with $\langle t, \I_1(\Cols_3)\rangle$. 
Clearly every solution of $I$ is a solution of $\chisub{}(I)$.  By Remark \ref{rem:bottop}, we will also assume that there are a variables $\bot$ and $\top$ such that every constraint tuple has $\bot$,$\top$ in the final two positions.

Now assume that $I\in \Y_{{\SEP}\cap2\text{-}\Rob}(\I_2(\Cols_3))$. We first show that $\chisub{}(I)$ is in $\Y_{2\text{-}\Rob}(\I_1(\Cols_3))$. Let $\{u,v\}$ be a pair of distinct elements in $V$ and consider an assignment $\alpha\colon \{u,v\}\to \{0,1\}$ that is locally compatible for $\chisub{}(I)$. 

\emph{Case}~$1$. If $\alpha$ is locally compatible for $I$, then we can use the fact that $I$ is in $\Y_{2\text{-}\Rob}(\I_2(\Cols_3))$ to extend $\alpha$ to a solution of $I$, which is also a solution of $\chisub{}(I)$. 

\emph{Case}~$2$: If $\alpha$ is not locally compatible for $I$, then $\alpha$ satisfies $(\alpha(u), \alpha(v))=(1,1)$ and in this case, we can extend $\alpha$ to a solution of $\chisub{}(I)$ by mapping every variable in $\chisub{}(V)$ to $1$. 

To see that $\chisub{}(I)\in \Y_{\SEP}(\I_1(\Cols_3))$, observe that solutions for $I$ also are solutions for $\chisub{}(I)$, and the assumption $I\in \Y_{\SEP}(\I_2(\Cols_3))$ shows that there are enough solutions to separate all pairs of variables.

We now show that if $I$ is a NO instance of $\II_2\text{-}\SAT$ then $\chisub{}(I)$ has only a trivial solution.  We prove the contrapositive.  Suppose $\chisub{}(I)$ has a non-trivial solution $\psi$ in $\II_1\text{-}\SAT$.  Hence $\psi(\bot)\neq \psi(\top)$.  Then $\psi$ is also a solution of $I$ in $\II_2\text{-}\SAT$.  Hence $I$ is not a NO instance of $\II_2\text{-}\SAT$. 
\end{proof}
%
The following is obtained immediately from Proposition~\ref{pro:GAPWSAT1} and Theorem~\ref{thm:GAP1&2}. 
\begin{thm}\label{thm:GAPII1}
Let $\Gamma$ be a constraint language on $\{0,1\}$ such that $\langle \Gamma \rangle=\II_1$. Then $\Gamma$ has the gap property $\GAP(\N_{\NTriv},\Y_{{\SEP} \cap(2,{\mathcal G})})$, for some finite set ${\mathcal G}$ of pp-formul{\ae} in the language of $\Gamma$, and consequently the problems $(2, {\mathcal G})\text{-}\Robust(\Gamma)$ and $\SEP(\Gamma)$ are $\bf{NP}$-complete.
\end{thm}

\subsection{The Boolean co-clone $\II_0$.}

Since the co-clone $\II_0$ is dual to $\II_1$, we obtain the following result via a symmetrical argument.  
%

\begin{thm}\label{thm:GAPII0}
Let $\Gamma$ be a constraint language on $\{0,1\}$ such that $\langle \Gamma \rangle=\II_0$. Then $\Gamma$ has the gap property $\GAP(\N_{\NTriv},\Y_{{\SEP}\cap(2,{\mathcal G})})$, for some finite set ${\mathcal G}$ of pp-formul{\ae} in the language of $\Gamma$, and consequently the problems $(2, {\mathcal G})\text{-}\Robust(\Gamma)$ and $\SEP(\Gamma)$ are $\bf{NP}$-complete.
\end{thm}

\subsection{The Boolean co-clone $\II$.}
By Proposition~\ref{pro:weakbase}, the following relation is a weak base for the co-clone $\II$. 
\[
\I(\Cols_{3})=
 \begin{bmatrix}
1 & 0 & 0 & 1 & 1 & 0 & 0 & 1 \\
0 & 1 & 0 & 1 & 0 & 1 & 0 & 1 \\
0 & 0 & 1 & 0 & 1 & 1& 0 & 1 \\
0 & 0 & 0 & 0 & 0 & 0 & 0 & 0 \\
1 & 1 & 1 & 1 & 1 & 1& 1 & 1
  \end{bmatrix}
\]
We will use $\mathbb{II}$ to denote the template $\langle \{0,1\};\I(\Cols_{3}) \rangle$ and $\II\text{-}\SAT$ to denote the constraint problem $\CSP(\mathbb{II})$.
We remark that instances in $\N_{\NTriv}(\I(\Cols_{3}))$ are exactly the instances of $\II\text{-}\SAT$ with at most two solutions. 

\begin{pro}\label{pro:GAPWSAT01}
The relation $\I(\Cols_{3}) $ satisfies $\GAP(\N_{\NTriv}, \Y_{{\SEP}\cap2\text{-}\Rob})$. 
\end{pro}
\begin{proof}
We show there is a polynomial-time reduction taking $\N_{\CSP}(\I(\Cols_{3}))$ to $\N_{\NTriv}(\I(\Cols_{3}))$ and taking instances in $\Y_{{\SEP}\cap2\text{-}\Rob}(\I_2(\Cols_{3}))$  to instances in $\Y_{{\SEP}\cap2\text{-}\Rob}(\I(\Cols_{3}))$. The final result is then obtained from Proposition~\ref{pro:WeirdoGAP} and Lemma~\ref{lem:Gap}.

Given an instance $I$ of $\II_2\text{-}\SAT$, construct an instance $\chisub{1}(I)$ of $\II\text{-}\SAT$ by replacing each constraint $\langle t=(v_1,\dots, v_8),\I_2(\Cols_3)\rangle$ in ${\mathcal C}$ with $\langle t, \I(\Cols_3)\rangle$. Clearly every solution of $I$ is a solution of $\chisub{1}(I)$.  We again call on Remark \ref{rem:bottop} to assume that $\bot$ and $\top$ record the final two positions of every constraint tuple.

Assume that $I$ is in $\Y_{{\SEP}\cap2\text{-}\Rob}(\I_2(\Cols_{3}))$. We first show that $\chisub{1}(I)$ is in $\Y_{2\text{-}\Rob}(\I(\Cols_{3}))$. Let $\{u,v\}$ be a pair of distinct elements in $V$ and let $\alpha\colon \{u,v\}\to \{0,1\}$ be an assignment that is locally compatible for $\chisub{1}(I)$.  

\emph{Case}~$1$: If $\alpha$ is locally compatible for $I$, then we can use the fact that $I\in \Y_{2\text{-}\Rob}(\I_2(\Cols_{3}))$ to extend $\alpha$ to a solution of $I$, which is also a solution of $\chisub{1}(I)$. 

\emph{Case}~$2$: If $\alpha$ is not locally compatible for $I$, then $\alpha$ satisfies $(\alpha(u), \alpha(v))=(0,0)$ or $(\alpha(u), \alpha(v))=(1,1)$. In the first case, we can extend $\alpha$ to a solution of $\chisub 1(I)$ by mapping every variable in $\chisub 1(V)$ to $0$, and in the second case we extend $\alpha$ by mapping every variable in $\chisub 1(V)$ to $1$. 

Now $\chisub 1(I)\in \Y_{\SEP}(\I(\Cols_{3}))$ follows immediately from the fact that $I$ is in $\Y_{\SEP}(\I_2(\Cols_{3}))$. 

Since every non-trivial solution of $\chisub 1(I)$ in $\II\text{-}\SAT$ separates $\bot$ from $\top$, it follows that a nontrivial solution of $\chisub 1(I)$ in $\II\text{-}\SAT$ is also a solution of $I$ in $\II_2\text{-}\SAT$.  In the contrapositive: if $I$ is a NO instance of $\II_2\text{-}\SAT$, then $\chi(I)$ has only trivial solutions. 
\end{proof}
The next result now follows from Proposition \ref{pro:GAPWSAT01} and Theorem~\ref{thm:GAP1&2}. 
\begin{thm}\label{thm:GAPII}
Let $\Gamma$ be a constraint language on $\{0,1\}$ such that $\langle \Gamma \rangle=\II$. Then $\Gamma$ has the gap property $\GAP(\N_{\NTriv},\Y_{{\SEP} \cap(2,{\mathcal G})})$, for some finite set ${\mathcal G}$ of pp-formul{\ae} in the language of $\Gamma$, and consequently the problems $(2, {\mathcal G})\text{-}\Robust(\Gamma)$ and $\SEP(\Gamma)$ are $\bf{NP}$-complete.
\end{thm}

\subsection{The Boolean co-clone $\IN$.}
By Proposition~\ref{pro:weakbase}, the following relation is a weak base for the co-clone $\IN$. 
 \[
 \N(\Cols_{3})=
 \begin{bmatrix}
 1 & 0 & 0 & 0 & 1 & 1 & 0 & 1 \\
 0 & 1 & 0 & 0 & 0 & 1 & 0 & 1 \\
 0 & 0 & 1 & 1 & 1 & 0 & 0 & 1 \\
 0 & 1 & 1 & 1 & 0 & 0 & 1 & 0 \\
 1 & 0 & 1 & 1 & 1 & 0 & 1 & 0 \\
 1 & 1 & 0 & 0 & 0 & 1& 1 & 0 \\
 0 & 0 & 0 & 0 & 0 & 0 & 0 & 0 \\
 1 & 1 & 1 & 1 & 1 & 1 & 1 & 1 \\
  \end{bmatrix}
\]
We will use $\mathbb{IN}$ to denote the template $\langle \{0,1\};  \N(\Cols_{3})\rangle$ and $\IN\text{-}\SAT$ to denote the constraint problem $\CSP(\mathbb{IN})$. 

\begin{pro}\label{pro:GAPBigWSAT01}
The relation $\N(\Cols_{3})$ satisfies $\GAP(\N_{\NTriv}, \Y_{{\SEP}\cap 2\text{-}\Rob})$. 
\end{pro}
\begin{proof}
We show there is a polynomial-time reduction taking NO instances of the constraint problem $\IN_2\text{-}\SAT$ to instances in $\N_{\NTriv}(\N(\Cols_3))$ and taking instances in $\Y_{{\SEP}\cap2\text{-}\Rob}(\N_2(\Cols_3))$ to instances in $\Y_{{\SEP}\cap2\text{-}\Rob}(\N(\Cols_3))$. The final result will then follow from Proposition~\ref{pro:BigWGAP} and Lemma~\ref{lem:Gap}.

Given an instance $I=(V;\{0,1\};\mathcal{C})$ of $\IN_2\text{-}\SAT$, construct an instance of $\chisub 2(I)$ of $\IN\text{-}\SAT$ by replacing each constraint $\langle t=(x_1,\dots, x_8),\N_2(\Cols_3)\rangle$ in ${\mathcal C}$ with $\langle t, \N(\Cols_3)\rangle$.  Clearly every solution of $I$ is a solution of $\chisub 2(I)$.  We will use the assumption that $\bot$ and $\top$ are the only variables that appear in the final two positions of a constraint tuple, which is justified by Remark~\ref{rem:bottop2}.

Now assume that $I$ belongs to $\Y_{{\SEP}\cap{2}\text{-}\Rob}(\N_2(\Cols_3))$. We first show that $\chisub 2(I)\in \Y_{2\text{-}\Rob}( \N(\Cols_3))$. Let $\{u,v\}$ be a pair of distinct elements in $V$ and consider an assignment $\alpha\colon \{u,v\}\to \{0,1\}$ that is locally compatible for $\chisub 2(I)$. 

\emph{Case}~$1$. If $\alpha$ is locally compatible for $I$, then we can use the fact that $I$ is in $\Y_{2\text{-}\Rob}(\N_2(\Cols_3))$ to extend $\alpha$ to a solution of $I$, which is also a solution of $\chisub 2(I)$. 

\emph{Case}~$2$. If the assignment $\alpha$ is not locally compatible for $I$, then $\alpha$ satisfies $(\alpha(u), \alpha(v))=(0,0)$ or $(\alpha(u), \alpha(v))=(1,1)$. In the first case, we can extend $\alpha$ to a solution of $\chisub 2(I)$ by sending every variable in $V$ to $0$ and in the second, we extend by sending every variable in $V$ to $1$. 

Now $\chisub 2(I)\in \Y_{\SEP}(\N(\Cols_3))$ follows from the fact that $I\in \Y_{\SEP}(\N_2(\Cols_3))$: we use the separating solutions of $I$, which are also solutions of $\chisub 2(I)$, to separate distinct pairs of elements of $\chisub 2(I)$.

Now assume that $\chisub 2(I)$ has a non-trivial solution $\psi$ in $\IN\text{-}\SAT$.  Then $\psi$ separates $\bot$ from $\top$, so is also a solution of $I$ in $\IN_2\text{-}\SAT$. In the contrapositive: if $I$ is a NO instance of $\IN_2\text{-}\SAT$, then $\chisub 2(I)$ has only trivial solutions for $\IN\text{-}\SAT$.
\end{proof}
The following result is a consequence of Proposition~\ref{pro:GAPBigWSAT01} and Theorem~\ref{thm:GAP1&2}.
\begin{thm}\label{thm:GAPIN}
Let $\Gamma$ be a constraint language on $\{0,1\}$ such that $\langle \Gamma \rangle=\IN$. Then $\Gamma$ has the gap property $\GAP(\N_{\NTriv},\Y_{{\SEP}\cap(2,{\mathcal G})})$, for some finite set ${\mathcal G}$ of pp-formul{\ae} in the language of $\Gamma$, and consequently the problems $(2, {\mathcal G})\text{-}\Robust(\Gamma)$ and $\SEP(\Gamma)$ are $\bf{NP}$-complete.
\end{thm}

\section{Towards a Trichotomy: proving tractablity}\label{sec:tractable}
 In this section, we investigate tractability for $\SEP(\Gamma)$ and $(2, {\mathcal F})$-$\Robust(\Gamma)$ over constraint languages $\Gamma$ on $\{0,1\}$.  We establish a theorem that covers all cases that are solvable in polynomial-time, see Theorem~\ref{thm:tractable} below. 
 
In the following it useful to recall that a finite relational structure (template) is a core if all of its unary polymorphisms are bijective.  

The following theorem appears in \cite[Proposition~$3.2$]{JACK}; strictly speaking, the result in \cite{JACK} holds for $k$-robust satisfiability, however the proof is applicable to the broader setting of $(k, {\mathcal F})$-robust satisfiability. We list only the properties required to establish our main results. 
\begin{thm}\cite[Proposition~$3.2$]{JACK}\label{thm:tract}
Let $\Gamma$ be a constraint language on a finite set $A$ and let ${\mathbb A}$ be the corresponding template. Consider ${\mathbb A}_0=\langle A; \Gamma \cup \{(a)\ |\ a\in A\}\rangle$.  There is a logspace Turing reduction from the following computational problems to $\CSP({\mathbb A}_0)$. 
\begin{itemize}
\item $(k, {\mathcal F})$-$\Robust({\mathbb A})$, for any $k\in \mathbb{N}$ and any finite set ${\mathcal F}$ of pp-formul{ae} in the language of $\Gamma$. 
\item $\SEP({\mathbb A})$. 
\end{itemize}
\end{thm}
Recall that to a core structure ${\mathbb A}$ we can add all singleton unary relations without increasing the complexity of $\CSP({\mathbb A})$, see \cite[Theorem~$4.7$]{BJK}. 
Thus, it follows from Theorem~\ref{thm:tract}, that in the case where ${\mathbb A}$ is a core, each of the problems listed is reducible to $\CSP({\mathbb A})$. 

\begin{thm}\label{thm:tractable}
Let $\Gamma$ be a constraint language on $\{0,1\}$. If $\IN \not\subseteq \langle \Gamma \rangle$, then the computational problems $\SEP(\Gamma)$ and $(2, {\mathcal F})$-$\Robust(\Gamma)$ are solvable in polynomial-time. 
\end{thm}
\begin{proof}
When $\IN \not\subseteq \langle \Gamma \rangle$, it follows from Post's co-clone lattice (see Figure \ref{clo:dualpost}), that $\IN \not\subseteq \langle \Gamma\cup\{(0),(1)\} \rangle$ and then it is known that the constraint problem $\CSP(\Gamma \cup\{(0),(1)\})$ is tractable; this can be found in Schaefer's original argument for example; see \cite[Lemma~4.1]{Schaef}.  Then from Theorem~\ref{thm:tract}, the problems $\SEP(\Gamma)$ and $(2, {\mathcal F})$-$\Robust(\Gamma)$ are solvable in polynomial-time. 
\end{proof}

\section{An algebraic approach}\label{sec:algebra}
The algebraic analysis of polymorphisms has been a very powerful tool in the study of $\CSP$  complexity; for the main development of this approach see~\cite{JCC, JCG, JCG97, JCP}.  In particular, the existence of polymorphisms satisfying certain equations has been shown to ensure tractability (see~\cite{Jeav98} for example) and can be used to identify appropriate algorithmic methods for solvability~\cite{LibKoz,fedvar}.  Conversely, the absence of polymorphisms satisfying certain equations ensures intractability~\cite{LTess}.  A first step in developing the algebraic theory was recognising that $\CSP$ complexity is preserved under the operator $\langle-\rangle$; this goes back to Schaefer \cite{Schaef}.  A second step was connecting the algebra of polymorphisms to fundamental concepts in universal algebra: first variety theory, where substructures, homomorphic images and direct products enable hardness results to be transferred from small domains to larger domains, and then the deeper analysis of tame congruence theory.  Much of the universal algebraic analysis is built on the work of~\cite{BJK}.  For a comprehensive overview of the algebraic approach see \cite{Barto,JKN13}.  

We have seen that for some computational problems the $\langle-\rangle$ operator appears too coarse. When the weaker operator $\langle-\rangle_{\not\exists}$ can be substituted, conventional polymorphism analysis can be replaced by the analysis of \emph{partial} polymorphisms.  In this section, we examine which universal algebraic techniques in the algebraic analysis of $\CSP$ complexity can be recovered for problems amenable to the $\langle-\rangle_{\not\exists}$ operator. More specifically, we obtain results for a number of variants of the $\CSP$ that allow to lift hardness results from Boolean domains to many problems on templates with non-Boolean domains.  The first of our results relates the complexity of a finite partial algebra to the complexity of its subalgebras and homomorphic images and the second allows to further restrict to idempotent partial polymorphisms.




Results of Bulatov, Jeavons and Krokhin provide the fundamental connection between CSP complexity and the tame-congruence theoretic analysis of polymorphism algebras (the polymorphisms of a template treated as operations of an algebra); see \cite[Theorems $5.2$, $5.4$]{BJK} or \cite[Theorem~$2.1$]{LTess} for a refined version.  These results link variety-theoretic properties, such as taking subalgebras and homomorphic images, to the complexity of the corresponding $\CSP$.  We now prove variations of this result, first in the context of $\GAP(\N_{\CSP}, \Y_{{\SEP}\cap(k, {\mathcal G})})$ (Theorem~\ref{thm:HSGap}) and then in the context of two other $\CSP$ variants, namely the equivalence problem and implication problem (Theorem~\ref{thm:HSEQUIV}).  The first variation will be established through a series of polynomial-time reductions, which we present in the form of lemmas. The constructions used for substructures (Lemma~\ref{lem:ParSubRed}) and homomorphisms (Lemma~\ref{lem:homredn}) are based on those in~\cite{BJK,LTess}, given in the standard CSP setting.  In the case of taking substructures, $\SEP$ and $(k, {\mathcal F})$-$\Robust$ carry through using the standard construction (subject to a change in the local compatibility condition ${\mathcal F}$).  The homomorphism case however requires proper amendment, including the addition of extra relations and non-trivial usage of the gap property.

For partial algebras there are a number of reasonable ways to define the concepts of subalgebra and homomorphism; see Gr\"atzer \cite{gra} for example.  
We take the definition of subalgebra coinciding with the model-theoretic notion of substructure and take the weakest (that is, most general) definition of homomorphism of the standard definitions.
\begin{defn}
Let $F$ be a set of finitary partial operation symbols. The partial algebra ${\mathbf B}=\langle B; F^{\mathbf B}\rangle$ is a \emph{subalgebra} of the partial algebra ${\mathbf A}=\langle A; F^{\mathbf A}\rangle$, if $B\subseteq A$ and for each $n$-ary $f\in F$, we have $\dom(f^{\mathbf B})={\dom(f^{\mathbf A})}\cap B^n$, and $f^{\mathbf B}$ agrees with  $f^{\mathbf A}$ on this set. 
\end{defn} 
\begin{defn}
Let $F$ be a set of finitary partial operation symbols and let ${\bf A}=\langle A;F^{\bf A}\rangle$ and ${\bf B}=\langle B;F^{\bf B}\rangle$ be partial algebras.  A map $\varphi\colon { A} \to {B}$ is a \emph{homomorphism} if, for every $n$-ary $f\in F$ and $(a_1, \dots, a_n)\in \dom(f^{\mathbf A})$, we have $(\varphi(a_1), \dots, \varphi(a_n))\in \dom(f^{\mathbf B})$ and $\varphi(f^{\mathbf A}(a_1, \dots, a_n))=f^{\bf B}(\varphi(a_1), \dots, \varphi(a_n))$. 
\end{defn}
We say that ${\mathbf B}$ is a \emph{homomorphic image} of ${\mathbf A}$ if there exists a surjective homomorphism $\varphi\colon { \bf A} \onto {\bf B}$. 
Note that homomorphic images are not uniquely determined by the map $\varphi\colon A \to B$.

The first result we require is a more general version of Theorem~\ref{thm:kRobustRed}. The addition of equality presents complications for $(2, {\mathcal F})$-robust satisfiability, and we require the use of a gap property to carry through the reduction.  The proof for $\SEP$ is similar, however a further slight variation of the proof is necessary.  We will say that an equality constraint $\langle(x_i,x_j),=_A\rangle$ is non-trivial if $i\neq j$, otherwise it is trivial. 

\begin{thm}\label{thm:weaksysGAP}
Let $A$ be a non-empty finite set and let $\Gamma$ be a constraint language over $A$. Let $R$ be any finite subset of $\langle \Gamma\rangle_{\not\exists}$ and let ${\mathcal F}$ be a finite set of pp-formul{\ae} in the language of $R$.  Then there is a polynomial-time Karp reduction from
$\SEP(R)$ to $\SEP(\Gamma)$\textup.
 Moreover, if $R$ has gap property $\GAP(\N_{\CSP}, \Y_{{\SEP}\cap(k, {\mathcal F})})$, then $\Gamma$ satisfies $\GAP(\N_{\CSP}, \Y_{{\SEP}\cap(k, {\mathcal G})})$, for some finite set ${\mathcal G}$ of pp-formul{\ae} in the language of $\Gamma$. 
 \end{thm}
 \begin{proof}
 First fix some NO instance $J$ of $\SEP(\Gamma)$, ensuring that if $\CSP(\Gamma)$ is nontrivial, then $J$ is a NO instance of $\CSP(\Gamma)$ (which is automatically a NO instance of $\SEP(\Gamma)$).
 
We now make a small adjustment to the proof of Lemma \ref{lem:CAEquivSolns}.  Given an $R$-instance $I$, we construct a $\Gamma$-instance $I'$.  Again, each constraint $\langle (v_1, \dots, v_\ell) , r^{\mathbb A}\rangle\in {\mathcal C}$ is to be replaced by constraints built from the conjuncts in the formula defining $r$ from $\Gamma$.  However, now there may be equality conjuncts.  

\emph{Case~$1$}.  If there is a non-trivial equality constraint $\langle(x_i,x_j),=_A\rangle$ introduced, then $I$ was a NO instance of $\SEP(R)$ as no satisfying solution can separate $x_i$ and $x_j$. Thus, for $I'$ we may output $J$.

\emph{Case~$2$}.  If all equality constraints constructed are trivial, then we may omit these from the conjuncts.  All of the conjuncts are now equality-free and so we are in the situation encountered in Lemma \ref{lem:CAEquivSolns}: the solutions of the instance $I$ are identical to those of $I'$.  
Hence $I$ is a YES instance of $\SEP(R)$ if and only if $I$ is a YES instance of $\SEP(\Gamma)$.  As this is trivially true in Case 1 also, we have the desired reduction.

Assume that $R$ satisfies $\GAP(\N_{\CSP}, \Y_{{\SEP}\cap(k, {\mathcal F})})$. It follows that $\CSP(R)$ is nontrivial and so $\CSP(\Gamma)$ is non-trivial as $\CSP(R)$ reduces to $\CSP(\Gamma)$.  Assume $I$ is a NO instance of $\CSP(R)$.  If Case~$1$ applies, then $I'=J$ is NO instance of $\CSP(\Gamma)$, as required.  If Case~$2$ applies, then $I'$ cannot have a solution.  So $I'$ is again a NO instance of $\CSP(\Gamma)$.

Now assume $I\in \Y_{{\SEP}\cap(k, {\mathcal F})}$.  Then Case $2$ applies and the solutions of the instance $I$ are identical to those of $I'$.  We are then essentially in the equality-free case considered in ($2$) and ($3$) of Theorem~\ref{thm:kRobustRed}.  Using the set $\mathcal{G}$ and argument given there, it follows that $I'$ is in $\Y_{{\SEP}\cap(k, {\mathcal G})}$. The final result is obtained by Lemma~\ref{lem:Gap}. 
 \end{proof}



In the following lemmas, it may be useful to recall that, for any non-empty set $A$ and set $F$ of partial operations on $A$, a $k$-ary relation $r$ on $A$ is in $\Inv(F)$ precisely when $r$ is a subuniverse of the partial algebra ${\langle A;F\rangle}^{k}$.
\begin{lem}\label{lem:homredn}
Let $F$ be a set of finitary partial operation symbols, let ${\bf A}=\langle A;F^{\bf A}\rangle$ and ${\bf B}=\langle B;F^{\bf B}\rangle$ be finite partial algebras and let $\varphi\colon {\bf A} \twoheadrightarrow {\bf B}$ be a surjective homomorphism.  Let $\Gamma$ be a finite subset of $\Inv(F^{\bf B})$, let ${\mathcal F}$ be a finite set of pp-formul{\ae} in the language of $\Gamma$ and let $k\in {\mathbb N}$.  Then for some finite set $\Gamma'\subseteq \Inv(F^{\bf A})$ there is a polynomial-time computable function reducing 
\begin{enumerate}[\quad \rm(1)]
\item $\CSP(\langle B; \Gamma \rangle)$ to $\CSP(\langle A; \Gamma' \rangle)$,
\item $(k, {\mathcal F})\text{-}\Robust(\langle B;\Gamma\rangle)$ to $(k, {\mathcal G})\text{-}\Robust(\langle A;\Gamma'\rangle)$, for some finite set ${\mathcal G}$ of pp-formul{\ae} in the language of $\Gamma'$, 
\end{enumerate}
and taking $\Y_{\SEP}(\Gamma)$ to $\Y_{\SEP}(\Gamma')$.
\end{lem}
\begin{proof}
We use a slight variation to the construction given in the proof of \cite[Theorem~$5.4$]{BJK}.   We begin by finding a subset $\Gamma'$ of $\Inv(F^{\bf A})$: for each $k$-ary relation $r$ in $\Inv(F^{\bf B})$, construct the following set
\[
\varphi^{\leftarrow}(r):=\{(a_1, \dots, a_k)\ |\ (\varphi(a_1), \dots, \varphi(a_k))\in r\} \subseteq A^k,
\]
and let ${\Gamma}'=\{ \varphi^{\leftarrow}(r)\ |\ r\in \Gamma\}\cup\{\ker(\varphi)\}$. Certainly ${\Gamma}'$ is finite and it is an easy exercise to check that $\ker(\varphi)$ is a subuniverse of ${\langle A; F^{\bf A}\rangle}^2$ and $\varphi^{\leftarrow}(r) \in \Inv(F^{\bf A})$, for each $r$ in $\Inv(F^{\bf B})$. 

Given an instance $I=(V;B;{\mathcal C})$ of $\CSP(\langle B; \Gamma \rangle)$, we construct an instance $I'=(V;A;{\mathcal C}')$ of $\CSP(\langle A; \Gamma' \rangle)$ by replacing each constraint $\langle s; r\rangle \in {\mathcal C}$ with $\langle s; \varphi^{\leftarrow}(r)\rangle$. Thus, ${\mathcal C}'=\{\langle s; \varphi^{\leftarrow}(r)\rangle\ |\ \langle s; r\rangle \in {\mathcal C}\}$.  

The proof for ($1$) is identical to the proof of \cite[Theorem~$5.4$]{BJK}. 
For ($2$), we begin by finding a set ${\mathcal G}$ of pp-formul{\ae} in the language of $\Gamma'$. For each $\rho(w_1, \dots, w_n)\in {\mathcal F}$ we construct a pp-formula $\rho_{\Gamma'}(w_1, \dots, w_n)$ in the language of $\Gamma'$ by replacing every occurrence of $r\in {\Gamma}$ with $\varphi^{\leftarrow}(r)\in\Gamma'$ and every occurrence of ${=_B}$ with $\ker(\varphi)$. Observe that $\rho_{\Gamma'}(a_1, \dots, a_n)$ is true in $\langle A;\Gamma'\rangle$ if and only if $\rho(\varphi(a_1), \dots, \varphi(a_n))$ is true in $\langle B;\Gamma\rangle$, for any $(a_1, \dots, a_n)\in A^n$ and $\rho(w_1, \dots, w_n)\in {\mathcal F}$.   Now assume that $I=(V;B;{\mathcal C})$ is a YES instance of $(k, {\mathcal F})\text{-}\Robust(\langle B;\Gamma\rangle)$ and consider an assignment $\alpha\colon S\to A$ with $|S|=k$. If $\alpha$ is ${\mathcal G}$ compatible for $I'$, then $\varphi\circ \alpha\colon S\to B$ is ${\mathcal F}$-compatible for $I$.  We can extend $\varphi\circ \alpha$ to a solution $\phi$ of $I$.  Then any function $f\colon V\to A$ such that $\varphi \circ f(v) = \phi(v)$, for all $v\in V$, is a solution of $I'$ that extends $\alpha$. For the converse direction, begin by assuming that $I'=(V;A;{\mathcal C}')$ is a YES instance of $(k, {\mathcal G})\text{-}\Robust(\langle A;\Gamma'\rangle)$ and consider an assignment $\beta\colon S\to B$ with $|S|=k$ that is ${\mathcal F}$-compatible for $I$. Then any $g\colon S\to A$ such that $\varphi\circ g(v)=\beta(v)$, for all $v\in V$ is ${\mathcal G}$-compatible for $I'$.  We can extend $g$ to a solution $\phi$ of $I'$.  Then $\varphi\circ \phi\colon V\to B$ is a solution of $I$ that extends $\beta$. 

To prove ($3$), assume that $I$ is a YES instance of $\SEP(\langle B; \Gamma \rangle)$. Let $\{v_1, v_2\}$ be a pair of distinct variables in $V$ and let $\phi\colon V\to B$ be a separating solution of $I$ for the pair $\{v_1,v_2\}$ guaranteed by the assumption.  Then any function $\alpha\colon V\to A$ such that $\varphi \circ \alpha(v) = \phi(v)$, for all $v\in V$, is a separating solution of $I'$ for the pair $\{v_1, v_2\}$.  Hence $I'$ is a YES instance of  $\SEP(\langle A; \Gamma' \rangle)$.   
\end{proof}

\begin{lem}\label{lem:ParSubRed}
Let $F$ be a set of finitary partial operation symbols, let ${\mathbf A}=\langle A; F^{\mathbf A}\rangle $ be a finite partial algebra and let ${\mathbf B}=\langle B; F^{\mathbf B}\rangle$ be a subalgebra of ${\mathbf A}$.  Let $\Gamma$ be a finite subset of $\Inv(F^{\mathbf B})$, let ${\mathcal F}$ a finite subset of pp-formul{\ae} in the language of $\Gamma$ and let $k\in{\mathbb N}$.  Then for some finite set $\Gamma'$ of $\Inv(F^{\mathbf A})$ there is a polynomial-time reduction from 
\begin{enumerate}[\quad \rm(1)]
\item $\CSP(\langle B; \Gamma \rangle)$ to $\CSP(\langle A; \Gamma' \rangle)$,
\item $\SEP(\langle B; \Gamma \rangle)$ to $\SEP(\langle A; \Gamma' \rangle)$, 
\item $(k, {\mathcal F})\text{-}\Robust(\langle B;\Gamma\rangle)$ to $(k, {\mathcal G})\text{-}\Robust(\langle A;\Gamma'\rangle)$, for some finite set ${\mathcal G}$ of pp-formul{\ae} in the language of $\Gamma'$. 
\end{enumerate}
\end{lem}
\begin{proof}
We first show that $\Inv(F^{\mathbf B})\subseteq \Inv(F^{\mathbf A})$.  Let $r$ be a $k$-ary relation on the set $B$ and assume that $r$ is invariant under $F^{\mathbf B}$.  We will show that $r$ is invariant under $F^{\bf A}$.  Let $f^{\mathbf A}\in F^{\mathbf A}$ be $n$-ary, and let $a_1=(a_{11}, \dots, a_{1n}), a_2=(a_{21}, \dots, a_{2n}), \dots, a_k=(a_{k1}, \dots, a_{kn})$ be tuples in $\dom(f^{\mathbf A})\subseteq A^{n}$.  Now assume that for all $i\in \{1, \dots n\}$, we have $(a_{1i}, a_{2i}, \dots a_{ki})\in r$.  Since $r$ is a relation on $B$, we have $\{a_1, \dots, a_k \}\subseteq\dom(f^{\mathbf B})=\dom(f^{\mathbf A}) \cap B^{n}$ and since $f^{\mathbf A}=f^{\mathbf B}$ on this set, it follows that 
\[
(f^{\mathbf A}(a_1), f^{\mathbf A}(a_2), \dots, f^{\mathbf A}(a_k))=(f^{\mathbf B}(a_1), f^{\mathbf B}(a_2), \dots, f^{\mathbf B}(a_k))\in r,
\]
and hence $r\in \Inv(F^{\mathbf A})$.
Now since ${\bf B}$ is a subalgebra of ${\bf A}$, the set $B$, viewed as a unary relation on $A$, belongs to $\Inv(F^{\mathbf A})$. So let ${\Gamma}'={\Gamma}\cup\{B\}\subseteq \Inv(F^{\bf A})$. 

The following construction is borrowed from \cite[Theorem~$5.2$]{BJK} and gives a reduction from $\CSP(\langle B; \Gamma \rangle)$ to $\CSP(\langle A; \Gamma' \rangle)$, proving ($1$).  Given an instance $I=(V;B;{\mathcal C})$ of $\CSP(\langle B;\Gamma\rangle)$, we construct an instance of $I'=(V;A;{\mathcal C}')$ of $\CSP(\langle A;\Gamma'\rangle)$ in the following way.  For each $v\in V$, add the constraint $\langle (v), B\rangle $ to ${\mathcal C}$. So $\mathcal C'=\mathcal C\cup \{(v)\in B\ |\ v\in V\}$.  Note that the solutions of $I=(V;B;{\mathcal C})$ are identical to those of $I'=(V;A;{\mathcal C}')$.  It follows that the separating solutions of $I$ are precisely the separating solutions of $I'$.  Hence $I$ is a YES instance of $\SEP(\langle B; \Gamma \rangle)$ if and only if $I'$ is a YES instance of  $\SEP(\langle A; \Gamma' \rangle)$.  This proves ($2$). 

For ($3$), assume that $I$ is a YES instance of $(k, {\mathcal F})$-$\Robust(\langle B; \Gamma \rangle)$ and let ${\mathcal G}={\mathcal F}\cup \{(x)\in B\}$.  Every ${\mathcal G}$-compatible assignment on $k$ variables of $I'$ preserves ${\mathcal F}$-constraints into $B$ as $\langle (v), B\rangle $ is a constraint of $I'$, for all $v\in V$. Thus, every ${\mathcal G}$-compatible assignment on $k$ variables for $I'$ is an ${\mathcal F}$-compatible assignment for~$I$, and conversely.  The reduction from $(k, {\mathcal F})\text{-}\Robust(\Gamma)$ to $(k, {\mathcal G})\text{-}\Robust(\Gamma')$ then follows from the fact that the solutions of $I=(V;B;{\mathcal C})$ are exactly the solutions of $I'=(V;A;{\mathcal C}')$. 
\end{proof}
We are now ready to prove the first main result of this section.  For any partial algebra ${\mathbf A}$ we let $\mathsf{HS}(\bf{A})$ be the smallest class of partial algebras in the same signature closed under the formation of homomorphic images ($\mathsf{H}$) and subalgebras ($\mathsf{S}$). 
\begin{thm}[$\mathsf{HS}$ Gap Theorem]\label{thm:HSGap}
Let ${\mathbb A}=\langle A; R\rangle$ and ${\mathbb B}= \langle B; S\rangle$ be templates, let ${\mathcal F}$ be a finite set of pp-formul{\ae} in the language of $R$ and let $k\in {\mathbb N}$. If $S$ has gap property $\GAP(\N_{\CSP}, \Y_{{\SEP}\cap(k, {\mathcal F})})$ and there exist partial algebras ${\mathbf A}=\langle A; F^{\mathbf A} \rangle$ and ${\mathbf B}=\langle B; F^{\mathbf B} \rangle$ such that 
\begin{enumerate}[\quad \rm(1)]
\item $F^{\mathbf B} \subseteq \pPol(S)$,
\item ${\mathbf B}\in \mathsf{HS}(\mathbf A)$, and 
\item $\pPol(R)\subseteq F^{\mathbb A}$,
\end{enumerate} 
then $R$ has gap property $\GAP(\N_{\CSP}, \Y_{{\SEP}\cap(k, {\mathcal G})})$.
\end{thm}
\begin{proof}
The result is established by carrying the gap property through items ($1$), ($2$) and ($3$). This requires showing there is a series of polynomial-time reductions preserving NO instances of $\CSP$ and instances that are simultaneously YES instances of $\SEP$ and $(k,{\mathcal F})$-$\Robust$.  

From ($1$), we have $\Inv(\pPol(S)) \subseteq \Inv(F^{\mathbf B})$ and so $S\subseteq \langle S \rangle_{\not\exists} \subseteq \Inv(F^{\mathbf B})$. From~($2$), there is a surjective homomorphism $\varphi\colon {\mathbf C} \twoheadrightarrow {\mathbf B}$, where ${\mathbf C}=\langle C; F^{\bf C}\rangle$ is a subalgebra of ${\mathbf A}$.  Thus, by Lemma~\ref{lem:homredn}, there is a finite subset $T$ of $\Inv(F^{\bf C})$ and a polynomial-time reduction taking instances in $\N_{\CSP}(S)$ to $\N_{\CSP}(T)$, and instances in $\Y_{{\SEP}\cap(k,{\mathcal F})}(S)$ to $\Y_{{\SEP}\cap(k,{\mathcal F}')}(T)$, for some finite set ${\mathcal F}'$ of pp-formul{\ae} in the language of $T$. Hence $T$ satisfies $\GAP(\N_{\CSP}, \Y_{{\SEP}\cap(k,{\mathcal F'})})$.  Now since ${\mathbf C}$ is a subalgebra of ${\mathbf A}$, by Lemma~\ref{lem:ParSubRed}, there is a finite set $U$ of $\Inv(F^{{\mathbf A}})$ and a polynomial-time reduction taking instances in $\N_{\CSP}(T)$ to $\N_{\CSP}(U)$, and instances in $\Y_{{\SEP}\cap(k,{\mathcal F'})}(T)$ to $\Y_{{\SEP}\cap(k,{\mathcal G}')}(U)$, for some finite set ${\mathcal G}'$ of pp-formul{\ae} in the language of $U$. Hence $U$ satisfies $\GAP(\N_{\CSP}, \Y_{{\SEP}\cap(k,{\mathcal G'})})$. From ($3$), it follows that $\Inv(F^{\mathbf A})\subseteq \Inv(\pPol(R^{\mathbb A}))=\langle R \rangle_{\not\exists}$. Thus, $U$ is a finite subset of $\langle R\rangle_{\not\exists}$.  By Theorem~\ref{thm:weaksysGAP}, it then follows that $R$ has gap property $\GAP(\N_{\CSP}, \Y_{{\SEP}\cap(k,{\mathcal G})})$, for some finite set ${\mathcal G}$ of pp-formul{\ae} in the language of $R$. 
\end{proof}

\begin{remark}
When ${\bf A}\neq{\bf B}$, Theorem~\ref{thm:HSGap} lifts the complete classification given by the Gap Trichotomy Theorem on Boolean domains to many problems on templates with non-Boolean domains. With further effort, direct products ($\mathsf{P}$) can be incorporated into condition ($2$) of Theorem~\ref{thm:HSGap}, however the argument is beyond the scope of the present article, and will appear in subsequent work.  For now we establish a weaker version. 
\end{remark}

We first introduce the following equivalence relations.
\begin{defn}
Let $A$ be a non-empty set and $\ell\in\mathbb{N}$. For any pair of elements $a_1=(a_{1,0}, \dots, a_{1,{\ell-1}}), a_2=(a_{2,0},\dots,  a_{2,{\ell-1}})\}$ in $A^{\ell}$ and for every $p,q\in\{0,\dots,{\ell-1}\}$, we define
\[
a_1\equiv_{pq} a_2\ \text{if}\ a_{1,p}=a_{2,q},
\]
\end{defn}
For any $\ell \in {\mathbb N}$, there are $\ell^2$ equivalence relations of the form $\equiv_{pq}$, for $p,q\in \{0,\dots,{\ell-1}\}$.
%


The basis for the following construction was inspired by \cite[Lemma~$2.4$]{LTess}, however the reduction requires significant modification with the main addition being a pre-processing ``local reflection'' procedure.
\begin{lem}\label{lem:PowerRed}
Let $F$ be a set of finitary partial operation symbols, let ${\mathbf A}=\langle A; F^{\mathbf A}\rangle$ be a finite partial algebra and consider ${\mathbf A}^{\ell}=\langle A^{\ell}; F^{{\mathbf A}^{\ell}}\rangle$. Let $\Gamma$ be a finite subset of $\Inv(F^{{\mathbf A}^{\ell}})$ and let ${\mathcal F}$ be a finite set of pp-formul{\ae} in the language of $\Gamma$.  Then for some finite set $\Gamma'\subseteq \Inv(F^{\mathbf A})$ there is a polynomial-time function taking
\begin{enumerate}[\quad \rm(1)]
\item $\N_{\CSP}(\Gamma)$ to $\N_{\CSP}(\Gamma')$, 
\item and $\Y_{(2, {\mathcal F})}(\Gamma)$ to $\Y_{\SEP}(\Gamma')$.
\end{enumerate}
\end{lem}
\begin{proof}
We begin by fixing the subset $\Gamma'$ of $\Inv(F^{\mathbf A})$.    A $k$-ary relation $r$ on $A^{\ell}$ can be considered as a ${\ell}k$-ary relation $r^{\flat}$ on $A$, by replacing each $k$-tuple of $\ell$-tuples 
\[
((a_{1,0}, \dots, a_{1,{\ell-1}}),\dots,(a_{k,0},\dots, a_{k,{\ell-1}}))
\]
by the ${\ell}k$-ary tuple
\[
(a_{1,0},\dots, a_{1,{\ell-1}},\dots ,a_{k,0}, \dots, a_{k,{\ell-1}}).
\]
In particular, if ${\bf A}=\langle A;F^{\bf A}\rangle$ is a partial algebra and $r$ is a subuniverse of the direct power $({\bf A}^{\ell})^k$, then $r^\flat$ is the corresponding subuniverse of $\mathbf{A}^{{\ell}k}$. Hence $r^\flat\in \Inv(F^{\mathbf A})$ and so let $\Gamma'=\{r^{\flat}\ |\ r\in\Gamma \}$. 

The reduction now proceeds in two steps. Given an instance $I=(V;A^{\ell};{\mathcal C})$ of $\CSP(\Gamma)$ where $V=\{v_i\ |\ i\leq|V|\}$, we begin by constructing an instance $I^{\flat}=(V^{\flat};A;{\mathcal C}^{\flat})$ of $\CSP(\Gamma'\cup\{=_{A}\})$ in the following way. 
\begin{enumerate}[\rm(i)]
\item Let $V^{\flat}=\{v_{i,p} \ |\ v_i\in V\ \text{and}\ p\in \{0,\dots,{\ell-1}\}\}$, and 
\item replace each constraint $\langle (v_1, \dots, v_k), r\rangle$ in ${\mathcal C}$ with 
\[
\langle (v_{1,0}, \dots, v_{1,{\ell-1}} \dots, v_{k,0},\dots ,v_{k,{\ell-1}}), r^{\flat}\rangle. 
\]
\end{enumerate}
At this point in the construction, we denote the instance by $I_0$.  To complete the construction of $I^{\flat}$, we apply the following ``local reflection'' procedure to $I$.  For every $i\leq |V|$, select some number $i'\neq i$ with $i'\leq |V|$.
For every $v_i \in V$ and $p,q\in \{0,\dots, {\ell-1}\}$ with $p\not=q$, 
\begin{itemize}
\item  if every ${\mathcal F}$-compatible assignment $\alpha\colon \{v_i,v_{i'}\}\to A^{\ell}$ maps 
maps $(v_i,v_i)$ into ${\equiv_{pq}}$, we add the constraint $\langle (v_{i,p},v_{i,q}), =_{A}\rangle$, 
\end{itemize}
and for every distinct pair $\{v_i, v_j\}$ in $V$ and $p,q\in \{0,\dots, {\ell-1}\}$ , 
\begin{itemize}
\item if every ${\mathcal F}$-compatible assignment $\alpha\colon \{v_i,v_j\}\to A^{\ell}$ maps $(v_i,v_j)$ into ${\equiv_{pq}}$, we add $\langle (v_{i,p},v_{j,q}), =_{A}\rangle$.
\end{itemize}

We now construct an instance $I^{\dflat}=(V^{\dflat};A;{\mathcal C}^{\dflat})$ of $\CSP(\Gamma')$, by eliminating all constraints involving $=_{A}$ in the usual way: fix some arbitrary total order $<$ on $V^{\flat}$ and let $\equiv_A$ denote the reflexive transitive closure of the relation ${=}_A$ on $V^\flat$.  For any $v_{i,p}$, let $\sigma({v_{i,p}})$ denote the $<$-smallest element of the ${\equiv}_{A}$-class of $v_{i,p}$.  Then $I^{\dflat}$ is simply $I^\flat$, but with each appearance of $v_{i,p}\in V^{\flat}$ replaced by $\sigma({v_{i,p}})$ (and all ${=}_A$ constraints removed). 

For ($1$), we prove the contrapositive. If $\varphi$ is a solution of $I^{\dflat}$ in $\CSP(\Gamma')$, then $\varphi\circ\sigma\colon V^{\flat}\to A$ is a solution of $I^{\flat}$ and therefore is also a solution of $I_0$. The map $\psi\colon V\to A^{\ell}$ defined by $\psi(v_i)=(\varphi(\sigma(v_{i,0})),\dots, \varphi(\sigma(v_{i,{\ell-1}})))$, for all $v_i\in V$, is a solution of $I$ in $\CSP(\Gamma)$.
 
For ($2$), assume that $I$ is a YES instance of $(2, {\mathcal F})$-$\Robust(\Gamma)$. We will show that $I^{\dflat}$ is a YES instance of $\SEP(\Gamma')$.  Let $\{v_{i,p}, v_{j,q}\}$ be a distinct pair of variables in $V^{\dflat}$. We consider two cases. 

\emph{Case~$1$}. If $i=j$ and $p\not=q$. Then there was some ${\mathcal F}$-compatible assignment $\alpha\colon \{v_i,v_{i'}\}\to A^{\ell}$ mapping $(v_i,v_i)$ into $A^{\ell}\setminus \equiv_{pq}$.  Hence $\pi_{p}(\alpha(v_i))\not=\pi_{q}(\alpha(v_i))$. We can extend $\alpha$ to a full solution $\varphi$ of $I$ in $\CSP(\Gamma)$ as we are assuming that $I$ is $(2, {\mathcal F})$-robustly satisfiable.  Then the map $\varphi^{\dflat}\colon V^{\dflat}\to A$ defined by $\varphi^{\dflat}(v_{st})=\pi_{t}(\varphi(v_s))$, for all $s\leq |V|$ and $t\in \{0,\dots, {\ell-1}\}$, is a solution of $I^{\dflat}$ that separates the pair $\{v_{i,p}, v_{i,q}\}$. 

\emph{Case~$2$}. If $i\not=j$ and $p,q\in \{0,\dots, {\ell-1}\}$.  Since $v_{i,p}$ and  $v_{j,q}$ were not identified by the above construction, there must have been some ${\mathcal F}$-compatible assignment $\alpha\colon \{v_{i},v_{j}\}\to A^{\ell}$ mapping $(v_i,v_j)$ into $A^{\ell}\setminus{\equiv_{pq}}$. Hence $\pi_{p}(\alpha(v_i))\not=\pi_{q}(\alpha(v_j))$. We can extend $\alpha$ to a solution $\varphi$ of $I$ in $\CSP(\Gamma)$.  Then the map $\varphi^{\dflat}\colon V^{\dflat}\to A$ defined by $\varphi^{\dflat}(v_{st})=\pi_{t}(\varphi(v_s))$, for all $s\leq |V|$ and $t\in \{0,\dots, {\ell}\}$, is a solution of $I^{\dflat}$ that separates the pair $\{v_{i,p}, v_{j,q}\}$.  
It follows from Case~$1$ and $2$ that $I^{\dflat}$ is a YES instance of $\SEP(\Gamma')$.
\end{proof}

%

\begin{thm}[$\mathsf{HSP}$ Gap Theorem]\label{thm:HSPGap}
Let ${\mathbb A}=\langle A; R \rangle$ and ${\mathbb B}= \langle B; S \rangle$ be templates and let ${\mathcal F}$ be a finite set of pp-formul{\ae} in the language of $R$. If $S$ has the gap property $\GAP(\N_{\CSP}, \Y_{{\SEP}\cap(2, {\mathcal F})})$ and there exist partial algebras ${\mathbf A}=\langle A; F^{\mathbf A} \rangle$ and ${\mathbf B}=\langle B; F^{\mathbf B} \rangle$ such that 
\begin{enumerate}[\quad \rm(1)]
\item $F^{\mathbf B} \subseteq \pPol(S)$,
\item ${\mathbf B}\in \mathsf{HSP}(\mathbf A)$, and 
\item $\pPol(R)\subseteq F^{\mathbb A}$,
\end{enumerate} 
then $R$ satisfies $\GAP(\N_{\CSP}, \Y_{\SEP})$.
\end{thm}
\begin{proof} 
The argument follows the proof of Theorem~\ref{thm:HSGap} up to and including the application of Lemma~\ref{lem:ParSubRed}, except now ${\mathbf C}$ is a subalgebra of ${\mathbf A}^{\ell}$, for some $\ell\in {\mathbb N}$.  We then apply Lemma~\ref{lem:PowerRed} before continuing on to the application of condition ($3$). The conclusion that $R$ has gap property $\GAP(\N_{\CSP}, \Y_{\SEP})$ is then established from the proof of Theorem~\ref{thm:weaksysGAP} (in contrast to Theorem~\ref{thm:weaksysGAP} directly).  
\end{proof}

We now observe that $\mathsf{HS}$ theorems can be obtained for other variants of the constraint satisfaction problem. The following \emph{equivalence problem} $\EQUIV$ and \emph{implication problem} $\IMPL$ have been studied in B{\"{o}}hler, Hemaspaandra, Reith and Vollmer~\cite{Bohl} and \cite{Schnoor2} respectively.  
\medskip

\noindent\fbox{
\parbox{0.95\textwidth}{\emph{Equivalence problem $\EQUIV({\mathbb A})$ over template ${\mathbb A}=\langle A; \Gamma^{\mathbb A}\rangle$.}\\
Instance: a pair of $\Gamma^{\mathbb A}$-instances ($I_1=(V;A;{\mathcal C}_1), I_2=(V;A;{\mathcal C}_2)$).\\
Question: for every assignment $\varphi$, is it true that $\varphi$ is a solution of $I_1$ if and only if $\varphi$ is a solution of $I_2$?
}
}
\\[1ex]
\noindent\fbox{
\parbox{0.95\textwidth}{\emph{Implication problem $\IMPL({\mathbb A})$ over template ${\mathbb A}=\langle A; \Gamma^{\mathbb A}\rangle$.}\\
Instance: a pair of $\Gamma^{\mathbb A}$-instances ($I_1=(V;A;{\mathcal C}_1), I_2=(V;A;{\mathcal C}_2)$).\\
Question: for every assignment $\varphi$, is it true that if $\varphi$ is a solution of $I_1$ then $\varphi$ is a solution of $I_2$?
}
}
\medskip

As observed in \cite[\S$6.2$]{Schnoor2}, the two problems are closely related, in fact they are equivalent under polynomial-time Turing reductions: a pair $(I_1,I_2)$ is a YES instance of $\IMPL({\mathbb A})$ if and only if $(I_1,(V;A;{\mathcal C}_1\cup \mathcal{C}_2))$ is a YES instance of $\EQUIV({\mathbb A})$, while $(I_1,I_2)$ is a YES instance of $\EQUIV({\mathbb A})$ if and only if both $(I_1,I_2)$ and $(I_2,I_1)$ are YES instances of $\IMPL({\mathbb A})$.

In the case of Boolean domains, both $\EQUIV$ and $\IMPL$ have been shown to experience $\mathbf{P}$ versus Co$\mathbf{NP}$ dichotomies with identical boundaries to those established for $\SEP$ and $(2, {\mathcal F})$-$\Robust$; see \cite{Bohl} for the original proof and \cite{Schnoor2} for a simplified proof using clone-theoretic techniques. The $\mathsf{HS}$ theorem for $\EQUIV$ and $\IMPL$ (Theorem~\ref{thm:HSEQUIV} below) lifts the complete classifications given by these dichotomies on Boolean domains to templates with non-Boolean domains.  This is particularly useful for $\EQUIV$ and $\IMPL$ as it is unknown if the operator $\langle -\rangle$ can be applied in the case of non-Boolean domains.  Even the operator $\langle -\rangle_{\not\exists}$ is unknown to be compatible with $\EQUIV$ in the non-Boolean case. The weaker operator $\langle -\rangle_{\not\exists,\not=}$ is known to be applicable (see Lemma~\ref{lem:impred} below), but this imposes the additional assumption of irredundancy. 

The proof of the $\mathsf{HS}$ theorem for $\EQUIV$ and $\IMPL$ relies on the following two results. It may be useful to recall the notion of irrendundancy given in Definition~\ref{defn:irrednt}.  We note that the use of gap properties is no longer required carry through the reductions.

\begin{lem}\cite[Proposition~3.11]{ISch}\label{lem:irred}
Let $R$ be a constraint language over a non-empty finite set $A$ and let $S$ be a finite subset of  $\langle R\rangle_{\not\exists}$. If $S$ is irredundant, then $S\subseteq \langle R\rangle_{\not\exists, \not=}$. 
\end{lem}

\begin{lem}\cite[Proposition~$2.2$]{Schnoor2}\label{lem:impred}
Let $R$ be a constraint language over a non-empty finite set $A$ and let $S\subseteq \langle R\rangle_{\not\exists, \not=}$.  Then there is polynomial-time reduction from $\EQUIV(S)$ to $\EQUIV(R)$ and from $\IMPL(S)$ to $\IMPL(R)$.
\end{lem}

\begin{thm}[$\mathsf{HS}$ Theorem for $\EQUIV$ and $\IMPL$]\label{thm:HSEQUIV}
Let ${\mathbb A}=\langle A; R\rangle$ and ${\mathbb B}= \langle B; S\rangle$ be templates. If $S$ is irredundant and there exist partial algebras ${\mathbf A}=\langle A; F^{\mathbf A} \rangle$ and ${\mathbf B}=\langle B; F^{\mathbf B} \rangle$ such that 
\begin{enumerate}[\quad \rm(1)]
\item $F^{\mathbf B} \subseteq \pPol(S)$,
\item ${\mathbf B}\in \mathsf{HS}(\mathbf A)$, and 
\item $\pPol(R)\subseteq F^{\mathbb A}$,
\end{enumerate} 
then there is a polynomial-time reduction from $\EQUIV(S)$ to $\EQUIV(R)$ and from $\IMPL(S)$ to $\IMPL(R)$.
\end{thm}
\begin{proof}
We prove the result only for $\IMPL$ as the proof for $\EQUIV$ is similar.  

Condition ($1$) implies that $\Inv(\pPol(S)) \subseteq \Inv(F^{\mathbf B})$ and so $S\subseteq \langle S \rangle_{\not\exists} \subseteq \Inv(F^{\mathbf B})$. Condition~($2$) gives a surjective homomorphism $\varphi\colon {\mathbf C} \twoheadrightarrow {\mathbf B}$, where ${\mathbf C}=\langle C; F^{\bf C}\rangle$ is a subalgebra of ${\mathbf A}$.  

We first show that for some $T\subseteq \Inv (F^{\mathbf C})$, the problem $\IMPL(\langle B; S \rangle)$ is reducible to $\IMPL(\langle C; T \rangle)$ in polynomial-time.  We use the construction given in the proof of Lemma~\ref{lem:homredn}, except now there is no need for the inclusion of the relation $\ker(\varphi)$.  Let $T=\{\varphi^{\leftarrow}(r)\ |\ r \in S\}\subseteq \Inv (F^{\mathbf C})$.  Given an instance 
\[
I=(I_1=(V;B;{\mathcal C}_1), I_2=(V;B;{\mathcal C}_1))
\]
of $\IMPL(\langle B; S \rangle)$, we construct an instance 
 \[
 I'=(I'_1=(V;C;{\mathcal C}'_1), I'_2=(V;C;{\mathcal C}'_2))
 \]
of $\IMPL(\langle C; T\rangle)$ by letting ${\mathcal C}'_i=\{\langle s; \varphi^{\leftarrow}(r)\rangle\ |\ \langle s; r\rangle \in {\mathcal C}\}$, for each $i\in \{1,2\}$. Now assume that $I$ is a YES instance of $\IMPL(\langle B; \Gamma \rangle)$.  Consider a solution $\psi\colon V\to C$ of $I'_1$. Then $\varphi \circ \psi\colon V\to B$ is a solution of $I_1$ and is therefore a solution of $I_2$.  Then any function $\mu\colon V\to C$ satisfying $\varphi \circ \mu =\varphi \circ \psi$ is a solution of $I'_2$. In particular, $\psi$ is a solution of $I'_2$.  Conversely, assume that $I'$ is a YES instance of $\IMPL(\langle C; \Gamma'\rangle)$ and consider a solution $\phi\colon V\to B$ of $I_1$. Let $\nu\colon V\to C$ be any function satisfying $\varphi \circ \nu =\phi$.  Then $\nu$ is a solution of $I'_1$ and therefore also a solution of $I'_2$. It then follows that $\varphi \circ \nu=\phi$ is a solution of $I_2$, as required. 

We now show that for some $U\subseteq \Inv ( F^{\mathbf A})$, there is a polynomial-time reduction from $\IMPL(\langle C; T \rangle)$ to $\IMPL(\langle A; U\rangle)$. In the degenerate case that ${\mathbf C}={\mathbf A}$, we let $U=T$.  Otherwise we use the construction given in the proof of Lemma~\ref{lem:ParSubRed}. Let $U=T\cup \{C\}$. Given an instance $I=(I_1, I_2)$ of $\IMPL(\langle C; T \rangle)$, we construct an instance $I'=(I'_1, I'_2)$ of $\IMPL(\langle A; U\rangle)$ in the following way: for each $i\in \{1,2\}$, we let $\mathcal C'_i=\mathcal C\cup \{(v)\in B\ |\ v\in V\}$. The reduction immediately follows from the fact that the solutions of $I_i$ are identical to those of $I'_i$, for each $i\in \{1,2\}$. 

Now from ($3$), we have that $\Inv(F^{\mathbf A})\subseteq \Inv(\pPol(R^{\mathbb A}))=\langle R \rangle_{\not\exists}$. Thus, $U$ is a finite subset of $\langle R\rangle_{\not\exists}$. Observe that the constructions applied to the constraint languages $S$ and $T$ above preserve irredundancy, and so $U$ is also irredundant.  By Lemma~\ref{lem:irred}, it follows that $U$ is a finite subset of $\langle R\rangle_{\not\exists, \not=}$.  Then Lemma~\ref{lem:impred} gives the final reduction from $\IMPL(\langle A; U\rangle)$ to $\IMPL(\langle A; R\rangle)$. 
\end{proof}

Another fundamental aspect of the algebraic approach to $\CSP$'s is the ability to restrict to so-called idempotent polymorphisms; see \cite[Theorem~$4.7$]{BJK}.  A partial polymorphism $f\colon\dom(f)\to\mathbb{A}$ is \emph{idempotent} if $f(a,\dots,a)=a$ for every $a\in A$ for which  $(a,\dots,a)\in \dom(f)$. As a finial result, we show that when analysing the complexity of $\SEP$ problems over non-Boolean domains we may also restrict to idempotent partial polymorphisms, further reducing the number of constraint languages required for classification.   An analogous result for the $(2, {\mathcal F})$-$\Robust$ problem can be obtained but is beyond the scope of the present article. 


We first fix some useful notation.  For a template ${\mathbb A}=\langle A; \Gamma\rangle$, let $V_A$ be a copy of $A$ given by $\{v_a\ |\ a\in A\}$ (to be treated as variables), and let $\diag({\mathbb A})$ denote the equality-free positive atomic diagram of ${\mathbf A}$:
\[
\diag({\mathbb A})=\{ \langle(v_{a_1}, \dots, v_{a_n}), r\rangle\ |\ (a_1, \dots, a_n)\in r,\ \text{where $r$ is an $n$-ary relation in $\Gamma$}\}.
\]

Let $\Gamma_{\Con}$ denote the union of $\Gamma$ with the set of all singleton unary relations $\{(a)\}$ for  $a\in A$.
The proof of the following theorem is based on examples described by Jackson \cite{jack2}. 
\begin{thm}\label{thm:unaryred}
Let $\Gamma$ be a constraint language over a finite non-empty set $A$.  If ${\mathbb A}=\langle A; \Gamma\rangle$ is a core relational structure, then $\SEP(\Gamma)$ is equivalent to $\SEP(\Gamma_{\Con})$, with respect to a polynomial-time Turing reduction. 
\end{thm} 

\begin{proof}
Every instance of $\SEP(\Gamma)$ is an instance of $\SEP(\Gamma_{\Con})$ as $\Gamma\subseteq \Gamma_{\Con}$, and is equivalently a YES instance (or NO instance) of both  problems.  Hence $\SEP(\Gamma)$ reduces to $\SEP(\Gamma_{\Con})$ in constant-time. 

We now show there is a polynomial-time Turing reduction from $\SEP(\Gamma_{\Con})$ to $\SEP(\Gamma)$. For a given instance $I=(V;A;{\mathcal C})$ of $\SEP(\Gamma_{\Con})$, we begin by constructing a new instance $I_A=(V';A;{\mathcal C}')$ of $\SEP(\Gamma_{\Con})$ in the following way: first let $V'$ be the disjoint union $V\cup V_A$ and construct ${\mathcal C}'$ by replacing every constraint $\langle v, \{(a)\}\rangle$ in ${\mathcal C}$ with $\langle (v,v_a), {=_{A}}\rangle$, and then taking the disjoint union with $\diag({\mathbb A})$. 

We will show there is a polynomial-time computable family ${\cat F}$ of $\Gamma$-instances satisfying the property that $I$ is a YES instance of $\SEP(\Gamma_{\Con})$ if and only if there exists $I'\in {\cat F}$ such that $I'$ a YES instance $\SEP(\Gamma)$. 

\emph{Case~$1$}.  If there exists $a\in A$ and constraints $\{\langle (v_1,v_a), {=_{A}}\rangle, \langle (v_2,v_a), {=_{A}}\rangle\}\subseteq {\mathcal C}'$ with $v_1\not=v_2$, then $I$ is automatically a NO instance of $\SEP(\Gamma_{\Con})$ as any satisfying solution of $I$ must identify $v_1$ and $v_2$.  In this case, we let ${\cat F}$ be the set containing a single fixed NO instance of $\SEP(\Gamma)$. 

\emph{Case~$2$}.  For every $a\in A$ at most one variable in $V$ is constrained to $\{(a)\}$. In this case there may be some $a\in A$ for which $\{(a)\}$ does not appear in any constraint in ${\mathcal C}$. Let $U$ be the set of all such ``unused" elements in $A$ and let $V_U$ be the set of all variables in $V$ that are not constrained to any singleton unary relation.  We first examine the number of ways to extend $\mathcal{C}'$ to include singleton constraints from $U$ without falling back into Case~$1$.   Each member of ${\cat F}$ will be constructed from such an extension.  For every subset $S\subseteq U$ and injective function $\iota$ from $U$ into $V_U$, let $I^{\iota}=(V';A;{\mathcal C}'\cup\{\langle (\iota(a),v_a), {=_{A}} \rangle\mid a\in S\})$.  Let $m:=|V_U|$ and let $n_i$ be the number of $i$-element subsets of $U$, for each $i\in \{0, \dots, |U|\}$.  Let $k=\max\{n_i\mid i=0,\dots,|U|\} \leq 2^{|A|}$.  The number of injective functions from an $i$ element set to $V_U$ is the $i$-permutation ${}^mP_i\leq m^{|A|}$.  Then the total number of instances included in ${\cat F}$ is 
\[
\sum_{i\in \{0, \dots, |U|\}} {}^{m}\!P_i \times n_{i}\leq (m+1)\times m^{|A|}\times 2^{|A|}.
\]
As $|A|$ is bounded by a constant and $m\leq|V|$, there is a uniform polynomial bound on the number of instances included in ${\cat F}$ (in fact, the total number of instances in ${\cat F}$ is in $O(|V|^{|A|+1}$)).

 
 Consider an extension $I^{\iota}=(V', A, {\mathcal C}'\cup\{\langle (\iota(a),v_a), {=_{A}} \rangle\mid a\in S\})$ of $I$. Apply the usual procedure to $I^{\iota}$ for eliminating all constraints involving the equality relation; the resulting variables will be denoted $V''$. When $\iota(a)=v$, it is notationally convenient to use both $v$ and $v_a$ to denote the same point of $V''$ after this identification: this is justified by the observation that in this construction there are no sequences of equalities of length more than one, so that distinct elements of $V$, as well as distinct elements of $V_A$, remain distinct after the identification.  In particular, this abuse of notation allows us to consider both $V_A$ and $V$ as (overlapping) subsets of~$V'$.  Let $\cat F$ be the set of these constructed instances. 

Assume there is some $I'\in \cat F$ such that $I'$ is a YES instance of $\SEP(\Gamma)$. So for every pair $\{v_1, v_2\}$ of distinct variables in $V''$ there is a separating solution $\varphi$ for that pair.  Since $\varphi$ preserves $\diag({\mathbb A})$, the restriction $\varphi\rest{V_A}$ determines an automorphism~$\alpha$ of ${\mathbb A}$, by way of $\alpha\colon a\mapsto b$ if $\varphi(v_a)=b$. Then the map $\alpha^{-1}\circ\varphi\rest{V}$ is a solution of $I$ separating the pair $\{v_1, v_2\}$. Hence $I$ is a YES instance of $\SEP(\Gamma_{\Con})$. 

Now assume that $I$ is a YES instance of $\SEP(\Gamma_{\Con})$.  We will show there exists some $I'$ in ${\cat F}$ such that $I'$ is a YES instance of $\SEP(\Gamma)$.  We first apply the following reflection procedure to $I_A$. For each $v\in V_U$, if every homomorphism from $I$ to $\langle A; \Gamma_{\Con}\rangle$ maps $v\in V_U$ to some $u\in U$, then we add the constraint $\langle (v,v_u), {=_{A}}\rangle$ to ${\mathcal C}'$.  Note that for distinct $v,v'\in V$ we cannot have both $\langle (v,v_u), {=_{A}}\rangle$ and $\langle (v',v_u), {=_{A}}\rangle$ added, as we are assuming there is a solution separating $v$ from $v'$.  Thus the instance constructed in this way coincides with one of the extensions from Case $2$. Find the $I'\in {\cat F}$ corresponding to this extension.  For every pair of distinct variables $\{v_1,v_2\}$ in $V$ there is a solution $\psi$ of $I$ separating $v_1$ and $v_2$.  Extend the separating solution $\psi$ for $\{v_1,v_2\}$ to a separating solution of $I'$ by sending $v_a$ to $a$.  Hence $I'$ is a YES instance of $\SEP(\Gamma)$. 
\end{proof}

\begin{defn}
Let $F$ be a set of finitary partial operation symbols and let ${\bf A}=\langle A ;F^{\bf A}\rangle$ be a finite partial algebra. We define the partial algebra
\[
{\bf A}_0:=\langle A; \id({[F^{\bf A}]}_p)\rangle,
\]
where $\id([F^{\bf A}]_p)$ consists of all idempotent operations from the strong partial clone generated by $F^{\bf A}$.  We refer to ${\bf A}_0$ as the \emph{full idempotent reduct of ${\bf A}$}. 
\end{defn}
Notice that $\Inv(\id([F^{\bf A}]_p)={\Inv(F^{\bf A})}_{\Con}$. The next result is an immediate consequence of Theorem~\ref{thm:unaryred}
\begin{cor}
Let ${\bf A}=\langle A; F^{\bf A}\rangle$ be a finite partial algebra and consider the full idempotent reduct ${\bf A}_0$ of ${\bf A}$. If ${\mathbb A}=\langle A; \Inv(F^{\bf A})\rangle$ is a core relational structure, then $\SEP(\Inv(F^{\bf A}))$ is equivalent to ${\SEP(\Inv(F^{\bf A})}_{\Con})$, with respect to a polynomial-time Turing reduction.
\end{cor}

\section{conclusion}
The results obtained demonstrate that gap properties can prove very powerful in the study of the complexity of constraint problems. In particular, the Gap Trichotomy Theorem provides dichotomies for an entire family of constraint problems on Boolean domains.  
Gap properties also facilitated an algebraic approach to the study of constraint problems amenable to partial polymorphism analysis: validity of most of the reductions relied intrinsically on the assumption of a gap property.  A possible future direction for work might be to extend this analysis to a tame-congruence theoretic approach. This would require a reworking of tame-congruence theory to adjust for partial algebras.  Could such an analysis lead to algebraic dichotomy conjectures for new variants of the constraint satisfaction problem? 

In future work, our focus will extend beyond ${\mathbf P}$ versus $\mathbf{NP}$ dichotomies to a finer-grained analysis of the complexity of the reductions given above.  While we did not develop these directions in the present article, arguments in the style of Larose and Tesson \cite{LTess} can be used to perform all of the reductions in log-space, and most using first order reductions.  Details for tightening the reductions will given in a subsequent paper, where we examine natural extensions of the $\SEP$ and ($2, {\mathcal F}$)-$\Robust$ conditions and their degree of compatibility with the operator~$\langle-\rangle$.  A gap dichotomy of the style given in the present article is believed to be true for these problems, however more technical machinery is required. 

\subsection*{Acknowledgements.}
The author is indebted to Marcel Jackson for many useful suggestions and discussions, as well as her PhD supervisors Brian Davey and Tomasz Kowalski, for their guidance and feedback. 

\bibliographystyle{plain}

\end{document}